\documentclass[journal,twocolumn]{IEEEtran}
\usepackage{amsfonts}    
\usepackage{color}
\usepackage{hyperref}
\usepackage{graphicx}
\usepackage{subcaption}
\usepackage{cite}
\usepackage{arydshln}
\usepackage{amsthm}
\usepackage{amsmath} 
\usepackage{amssymb}  
\usepackage[boxed,ruled,lined]{algorithm2e}                             

\usepackage{amsfonts}    
\usepackage{color}
\usepackage{graphicx}
\usepackage[dvips]{epsfig}
\usepackage{graphics}
\usepackage{cite}
\usepackage{arydshln}
\usepackage{amsmath} 
\usepackage{amssymb}  
\usepackage[boxed,ruled,lined]{algorithm2e}

\def\twon #1{\|#1\|}
\def\rainfty{\rightarrow\infty}
\def\ra{\rightarrow}

\def\argmin{\text{argmin}}

\def\bE{\mathbb{E}}

\def\bP{\mathbb{P}}
\def\cQ{\mathbb{Q}}
\def\bR{\mathbb{R}}

\def\cE{\mathcal{E}}
\def\cF{\mathcal{F}}
\def\cG{\mathcal{G}}
\def\cH{\mathcal{H}}

\def\cL{\mathcal{L}}

\def\cN{\mathcal{N}}

\def\cP{\mathcal{P}}
\def\cQ {\mathcal{Q}}

\def\cU{\mathcal{U}}
\def\cV{\mathcal{V}}

\def \qed {\hfill \vrule height6pt width 6pt depth 0pt}
\def\bee{\begin{equation}}
\def\ene{\end{equation}}
\def\beq{\begin{eqnarray}}
\def\enq{\end{eqnarray}}

\newtheorem{assum}{Assumption}
\newtheorem{lem}{Lemma}
\newtheorem{rem}{Remark}
\newtheorem{cor}{Corollary}
\newtheorem{thm}{Theorem}

\newtheorem{prob}{Problem}
\newtheorem{exmp}{Example}

\newtheorem{defi}{Definition}

\begin{document}
\title{Distributed Algorithms for Robust Convex Optimization via the Scenario Approach}
\author{Keyou You,~\IEEEmembership{Member,~IEEE}, Roberto Tempo, \IEEEmembership{Fellow,~IEEE}, and Pei Xie  \thanks{This work was supported by the National Natural Science Foundation of China (61304038),  Tsinghua University Initiative Scientific Research Program.} \thanks{Keyou You and Pei Xie are with the Department of Automation and TNList, Tsinghua University, 100084, China (email: youky@tsinghua.edu.cn, xie-p13@mails.tsinghua.edu.cn).} \thanks{Roberto Tempo  is deceased and was with CNR-IEIIT, Politecnico di Torino, Torino, 10129, Italy.}}

\maketitle
\begin{abstract}
This paper proposes distributed algorithms to solve robust convex optimization (RCO) when the constraints are affected by nonlinear uncertainty. We adopt a scenario approach by randomly sampling  the uncertainty set. To facilitate the computational task,
instead of using a single centralized processor to obtain a ``global solution" of the
scenario problem (SP), we resort to {\it multiple interconnected processors} that are distributed among different nodes of a network to simultaneously solve the SP. Then, we propose a primal-dual sub-gradient algorithm and a random projection algorithm to distributedly solve the SP over undirected and directed graphs, respectively. Both algorithms are given in an explicit recursive form with simple iterations, which are especially suited for processors with limited computational capability.  We show that, if the underlying graph is strongly connected, each node asymptotically computes a common optimal solution to the SP with a convergence rate $O(1/(\sum_{t=1}^k\zeta^t))$ where $\{\zeta^t\}$ is a sequence of appropriately decreasing stepsizes. That is, the RCO is effectively solved in a distributed way. The relations with the existing literature on robust convex programs are thoroughly discussed and an example of robust system identification is included to validate the effectiveness of our distributed algorithms.
\end{abstract}

\section{INTRODUCTION}
A robust convex optimization (RCO) is a convex optimization problem where an infinite number of constraints are parameterized by uncertainties. This problem has found wide applications in control analysis and synthesis of complex systems, as well as in other areas of engineering \cite{calafiore2006scenario,tempo2013randomized}. As the dependence of the constraints on the uncertainties may be nonlinear, RCO is generally not easily solvable. In fact, the study of RCO bears a vast body of literature, see e.g. \cite{ben1998robust,scherer2005relaxations,bertsimas2011theory} and references therein.

In this paper, we adopt  a {\em scenario approach}, which was
first introduced in \cite{calcammp, calafiore2006scenario} to solve RCO. In particular, we randomly sample the uncertainty set and obtain a standard convex optimization called the scenario problem (SP). The guarantees of optimality are then given in a probabilistic sense and an explicit bound on the probability that the original constraints are violated is provided. The striking feature of this approach is that the sample complexity, which guarantees that a solution to the SP is optimal with a given level of confidence, can be computed a priori. We also refer to \cite{tempo2013randomized,calafiore2011research} for general properties and specific randomized algorithms to cope with uncertainty in systems and control.


To facilitate the computational task, instead of using a single processor to solve the SP, this paper proposes a distributed computing framework with many interconnected processors.
The challenging problem is to distribute the computational task among the nodes of a network, each representing a single processor. The idea is to break a (possibly) large number of constraints of the SP into many small sets of {\em local} constraints that can be easily handled in each node. That is, each node computes some optimal solution of the SP with a low computational cost. Under local interactions between nodes, the SP is then collaboratively solved in every node via three key steps.

First, every node randomly samples the uncertainty set of RCO, with the sample size inversely proportional to the total number of nodes or being a priori determined by its computational capability. Although this idea has been adopted in \cite{carlone2014distributed, notarstefano2011distributed} to solve the SP, our approach is substantially different. In particular, after sampling,  each node in \cite{carlone2014distributed} requires to completely solve a local SP at each iteration and exchange the set of active constraints with its neighbors. The process continues until a consensus on the set of active constraints  is reached. Finally, every node solves its local SP under {\em all} active constraints of the SP.  Clearly, the number of constraints in every local SP increases with the number of iterations.  In some extreme cases, each constraint in the SP can be active, and every node eventually solves a local SP that has the same number of constraints as the SP. Thus,  the computational cost in each node is not necessarily reduced.
Moreover, each node cannot guarantee to obtain the same optimal solution to the SP. Since an active constraint may become inactive in any future iteration, identifying the active constraints cannot be recursively implemented, and this computation is very sensitive to numerical errors. On the contrary, each node in this paper only needs to handle a fixed number of local constraints and recursively run an explicit algorithm with very simple structure.

Second, the SP is reformulated as a distributed optimization problem with many decoupled small sets of local constraints and a coupled constraint, which is specially designed in conformity with the network structure. 
If the number of nodes is large, each node only needs to deal with a very small number of local constraints. The information is then distributed across the network via the coupled constraint, so that it can be locally handled. We recall that a similar technique has been already adopted to solve distributed optimization problems, see e.g. \cite{boyd2011distributed,nedich2015convergence}, which are only focused on convex optimization problems and no robustness issues are addressed. On the other hand, robust optimization has also attracted significant attention in many research areas \cite{ben2009robust,gorissen2015practical}, but the proposed approaches are fully centralized. In this paper, we address both distributed and robust optimization problems simultaneously.

Third, each node of the network keeps updating a local copy of an optimal solution by individually handling its local constraints and interacting with its neighbors to address the coupled constraint. If the graph is strongly connected, every pair of nodes can indirectly access information from each other. To this purpose,  we develop two recursive distributed algorithms for each node to interact with the neighbors to solve the SP by utilizing the constraint functions under undirected and directed graphs, respectively. For both algorithms, the computational cost per iteration only involves a few additions and multiplications of vectors, in addition to the computation of the sub-gradients of parameterized constraint functions. Thus, the computational cost is small in each node, and the approach is particularly useful for solving a large-size optimization problem with many  solvers of reduced power.

For undirected graphs, where the information flow between the nodes is bidirectional, we solve the distributed optimization problem by using an augmented Lagrangian function with a quadratic penalty \cite{bertsekas1999nonlinear}. Following this approach, a distributed primal-dual sub-gradient algorithm is designed to find a saddle point.  In this case, both the decoupled and coupled constraints are handled by introducing Lagrange multipliers, which provide a natural approach from the optimization viewpoint. For the coupled constraint, each node also needs to broadcast its estimate of an optimal solution to the SP, and the modified Lagrange multipliers to the neighbors, after which it recursively updates them by jointly using sub-gradients of local constraint functions. We show that each node finally converges to some common optimal solution to the SP. We remark that most of the existing work on distributed optimization \cite{chatzipanagiotis2015augmented,lee2016asynchronous,nedich2015convergence} uses the Euclidean projection to handle  local constraints. The projection is easy to perform only if the projection set has a special structure, which is generally not the case in the SP. From this perspective, our algorithm is more attractive to solve the SP problem in the context of distributed algorithms.

For directed graphs, the information flow between nodes is unidirectional and the primal-dual algorithm for undirected graphs cannot be used. To overcome this issue, we address the coupled constraint by adopting a consensus algorithm and design a novel two-stage recursive algorithm.  At the first stage, we solve an unconstrained optimization problem which removes the decoupled local constraints in the reformulated distributed optimization and obtain an intermediate state vector in each node. We notice that, in  the classical literature \cite{nedic2009distributed,duchi2012dual,gharesifard2014distributed,lee2016asynchronous}, the assumption on balanced graphs is often made. In our paper, this restrictive assumption is removed and this step is non-trivial, see e.g. \cite{nedic2015distributed,xi2015directed}.
At the second stage, each node individually addresses its decoupled local constraints by adopting a generalization of Polyak random algorithm \cite{nedic2011random}, which moves its intermediate state vector toward a randomly selected local constraint set. Combining these two stages, and under some mild conditions, both consensus and feasibility of the iteration in each node are achieved almost surely.  Although this distributed algorithm is completely different from the primal-dual sub-gradient algorithm previously described, both algorithms essentially converge at a speed $O(1/(\sum_{t=1}^k\zeta^t))$ where $\{\zeta^t\}$ is a sequence of appropriately decreasing stepsizes.

The rest of this paper is organized as follows. In Section \ref{sec_rcp}, we formulate RCO and include four motivating examples, after which the probabilistic approach to RCO is introduced. In Section \ref{sec_nps}, we describe a  distributed computing framework for the SP. In Section \ref{sec_npsa}, a distributed algorithm is proposed via the primal-dual sub-gradient method for undirected graphs and show its convergence. In Section \ref{sec_nrpa},  we design a distributed random projected algorithm over directed graphs to solve RCO. An example focused on robust system identification is included in Section \ref{sec_simulation}. Some brief concluding remarks are drawn in Section \ref{sec_conclusion}. 

A preliminary version of this work appeared in \cite{you2016parallel}, which only addresses undirected graphs with a substantially different approach. This paper provides significant extensions to directed graphs using randomized algorithms, establish their convergence properties, include the complete proofs and provide new simulation results for robust system identification.

{\bf Notation:} The sub-gradient  of a vector function $y =[y_1,\ldots,y_n]'\in\bR^n$ 
whose components are convex functions with respect to an input vector $x\in\bR^m$ is denoted by $\partial y=[\partial y_1, \ldots, \partial y_n]'\subseteq\bR^{n\times m}$. For two non-negative sequences $\{a^k\}$ and $\{b^k\}$, if there exists a positive constant $c$ such that $a^k\le c \cdot b^k$, we write $a^k=O(b^k)$. For two vectors $a=[a_1,\ldots,a_n]'$ and $b=[b_1,\ldots,b_n]'$, the notation $a\succeq b$ means that $a_i$ is greater than $b_i$ for any $i\in\{1,\ldots,n\}$. A similar notation is used for $\succ$, $\preceq$ and $\prec$. The symbol ${\bf 1}$ denotes the vector with all entries equal to one. Given a pair of real matrices of suitable dimensions, $\otimes$ indicates their Kronecker product. Finally, $f(\theta)_{+}=\max\{0,f(\theta)\}$ is the positive part of $f$,
$\rm Tr(\cdot)$  is the trace of a matrix and $\|\cdot\|$ denotes Euclidean norm.

\section{Robust Convex Optimization and Scenario Approach}
\label{sec_rcp}
\subsection{Robust Convex Optimization}
Consider a robust convex optimization (RCO) of the form
\beq
\min_{\theta\in\Theta}&&\hspace{-0.5cm}c'\theta~~\text{subject to}~f(\theta,q)\le 0, \forall q\in\cQ ,\label{rcp_opt}
\enq
where $\Theta\subseteq \bR^n$ is a convex and closed set with non-empty interior, and the scalar-valued function $f(\theta,q): \bR^n\times \cQ  \rightarrow \bR$ is convex in the decision vector $\theta$ for any $q\in\cQ  \subseteq \bR^l$. The uncertainty $q$ enters into the constraint function $f(\theta,q)$ without assuming any structure, except for the Borel measurability  \cite{ash2000pam} of $f(\theta,\cdot)$ for any fixed $\theta$. In particular, $f(\theta,\cdot)$ may be affected by parametric (possibly nonlinear) and nonparametric uncertainty.

Note that a linear objective function is not essential and the results of the paper still hold for any convex function by a simple relaxation. Specifically, consider a convex objective function $f_0(\theta)$ and introduce an auxiliary variable $t$. Then, the optimization in (\ref{rcp_opt}) is equivalent to
$$
\min_{\theta\in\Theta,t\in\bR}t~~\text{subject to}~f_0(\theta)- t\le 0~\text{and}~f(\theta,q)\le 0, \forall q\in\cQ.
$$

Obviously, the above objective function becomes linear in the augmented decision variable $(\theta,t)$ and is of the same form as (\ref{rcp_opt}). That is, there is no loss of generality to focus on a linear objective function.

\subsection{Motivating Examples}
The robust convex optimization in (\ref{rcp_opt}) is  crucial in many areas of research, see e.g. \cite{ben2009robust,bertsimas2011theory} and references therein for more comprehensive examples. Here we present some important applications for illustration.
\begin{exmp}[Robust MPC] Consider uncertain linear systems
\bee
x^{k+1}=A(q)x^k+B(q)u^k\label{system}
\ene
where $q\in\cQ$ represents the system uncertainty. The robust model predictive control (MPC) aims to solve the following optimization problem
\beq
&&\min_{u^k,\ldots,u^{k+h-1}}\max_{q\in\cQ} \sum_{j=k}^{k+h-1}g(x^j,u^j)+v(x^{k+h})\nonumber\\
&&\text{subject to}~u^j, \ldots, u^{k+h-1}\in\cU~\text{and}~(\ref{system}),\nonumber
\enq
where $g$ and $v$ are convex functions, and $\cU$ is convex and closed.  Let $\theta=(u^k,\ldots,u^{k+h-1})$,  it follows from (\ref{system}) that the objective function can be rewriten as $ J(\theta,q):=\sum_{j=k}^{k+h-1}g(x^j,u^j)+v(x^{k+h})$. Hence, the robust MPC is reformulated as the following RCO $$\min_{ \eta, \theta\in\cU^h} \eta~\text{subject to}~J(\theta,q)-\eta\le 0, \forall q\in\cQ.$$\end{exmp}
\begin{exmp}[Distributed robust optimization]\label{exmp_dro}
Consider the distributed robust optimization problem
\bee\label{opt_distrcp}
\min_{\theta\in\Theta}~\sum_{j=1}^m f_j(\theta,q_j),
\ene
where $f_j$ is only known to node $j$ and $q_j\in\cQ_j$ represents the uncertainty in node $j$ and its bounding set. Moreover, $f_j(\theta,q_j)$ is convex in $\theta$ for any $q_j$ and is Borel measurable in $q_j$ for any fixed $\theta$.

From the worst-case point of view, we are interested in solving the following optimization problem
\bee
\min_{\theta\in\Theta}~ \sum_{j=1}^m \left(\max_{q_j\in\cQ_j}  f_j(\theta,q_j)\right). \label{opt_robust}
\ene

However, the uncertainty $q_j$ generically enters the objective function $f_j(\theta,q_j)$ in (\ref{opt_distrcp}) without any specific structure, so that the objective function cannot be explicitly found.
To solve (\ref{opt_robust}), we note that it is equivalent to the following optimization problem
\bee
\min_{\theta\in\Theta, t}~ \sum_{j=1}^m t_j~\text{subject to}~\max_{q_j\in\cQ_j}  f_j(\theta,q_j)- t_j \le 0, \forall j\in\cV.  \label{opt_scenarioex}
\ene

Let $f(t,\theta,q)=[f_1(\theta,q_j)- t_1,\ldots,f_m(\theta,q_m)- t_m]'$ where $t=[t_1,\ldots,t_m]'$ and $q=[q_1,\ldots,q_m]'$ and $\cQ=\cQ_1\times \ldots \times \cQ_m$. Then, the optimization in (\ref{opt_scenarioex}) is equivalent to
\bee
\min_{\theta\in\Theta, t}~ \sum_{j=1}^m t_j~\text{subject to}~f(t,\theta,q)\preceq 0, \forall q\in\cQ.\label{opt_scenarioex1}
\ene
\end{exmp}

Clearly, (\ref{opt_scenarioex1}) is RCO of the form in (\ref{rcp_opt}), except that $f_j$ is only known to node $j$. However, this is not an issue as discussed in Example \ref{rem_dro} in Section III-B.


\begin{exmp}[LASSO] Consider the least squares (LS) problem
\bee
\min\limits_v \twon{b-Xv},\nonumber
\ene
where $X\in\bR^{l\times n}$ is the regression matrix and $b$ is the measurement vector.
It is well-known that the LS solution has poor numerical properties when the regression matrix is {\em ill-conditioned}. A common approach for addressing it is to introduce $\ell^1$ {\em regularization} technique, which results in a LASSO problem
\bee
\min\limits_v \{\twon{b-Xv}+\sum_{i=1}^nc_i|v_i|\}, \nonumber
\ene
where $c_i>0$ quantifies the robustness of the solution with respect to the $i$-th column of $X$. By \cite{xu2010robust}, the LASSO is in fact equivalent to a robust LS problem
\bee\label{robustls}
\min\limits_v\max\limits_{q\in\cQ}\twon{b-(X+q)v}
\ene
with the following uncertainty set
$$
\cQ=\{[q_1,\ldots,q_n]|\twon{q_j}\le c_j,j=1,\ldots,n\}.
$$
\end{exmp}
From (\ref{robustls}), the LASSO is inherently robust to the uncertainty in the regression matrix $X$, and the weight factor $c_i$ quantifies its robustness performance.  Note that the optimization in (\ref{robustls}) can also be reformulated as RCO in (\ref{rcp_opt}).

\begin{exmp}[Distribution-free robust optimization] \label{exa_dfro}Consider a distribution-free robust optimization under moment constraints
\bee
\min_{\theta\in\Theta}\max_{{q}\in\cP}\bE[f(\theta,{q})]\label{dfro}
\ene
where $f(\theta,q)$ is a utility convex function in the decision variable $\theta$ for any given realization of the random vector $q$, and the expectation $\bE[\cdot]$ is taken with respect to ${q}$. Moreover, $\cP$ is a collection of random vectors with the same support, first- and second-moments
$$
\cP=\{{q}: \text{\rm supp}({q})=\cQ, \bE[{q}]=\mu,\bE[{q}{q}']=\Sigma\}.
$$

In light of \cite{delage2010distributionally} and the duality theory \cite{shapiro2001duality}, the optimization problem (\ref{dfro}) is equivalent to  RCO
\beq
&&\min_{\theta,\alpha,\beta,\Omega}\{\alpha+\mu'\beta+ {\rm Tr
(\Omega'\Sigma)}\}\nonumber\\
&&\hspace{-1cm}\text{subject to}~\theta\in\Theta, \alpha+q'\beta+q'\Omega q\ge f(\theta,q),\forall q\in\cQ,\nonumber.
\enq
Clearly, the optimization (\ref{dfro}) is reformulated as RCO of the same form as (\ref{rcp_opt}).
\end{exmp}

Although the stochastic programming (\ref{dfro}) is a convex optimization problem, one must often resort to Monte Carlo sampling to solve it, which is computationally challenging, as it may also need to find an appropriate sampling distribution. Unless $f$ has a special structure, it is very difficult to obtain such a distribution \cite{barmish1997uniform}. In the next section, we show how RCO can be effectively solved via a {\em scenario approach}.

\subsection{Scenario Approach for RCO}
\label{sec_sarco}

The design constraint $f(\theta,q)\le 0$ for all possible $q\in\cQ$ is crucial in the study of robustness of complex systems, e.g.  $\cH_{\infty}$ performance of a system affected by the parametric uncertainty and  the design of uncertain model predictive control \cite{petersentempo}. However, obtaining worst-case solutions has been proved to be computationally difficult, even NP-hard as the uncertainty $q$ may enter into $f(\theta,q)$ in a nonlinear manner. In fact, it is generally very difficult to explicitly characterize the constraint set with uncertainty, i.e.,
\bee
\{\theta|f(\theta, q)\le 0,\forall q\in\cQ \}, \label{rob_const}
\ene
which renders it impossible to directly solve RCO.  There are only few cases when the uncertainty set is tractable \cite{ben2009robust}. Furthermore, this approach introduces undesirable {\em conservatism}. For these reasons, we adopt the scenario approach.

Instead of satisfying the hard constraint in (\ref{rob_const}), the idea of this approach is to derive a probabilistic {\em approximation} by means of a finite number of random constraints, i.e,
\bee
\bigcap_{i=1}^{N_{bin}}\{\theta|f(\theta,q^{(i)})\le 0\}\label{rob_rand}
\ene
where $N_{bin}$ is a positive integer representing the constraint size, and $\{q^{(i)}\}\subseteq \cQ$ are independent identically distributed (i.i.d.) samples extracted according to an arbitrary absolutely continuous (with respect to the Lebesgue measure) distribution $\bP_q(\cdot)$ over $\cQ$.

Regarding the constraint in (\ref{rob_rand}), we only guarantee that most, albeit not all, possible uncertainty constraints in RCO are not violated. Due to the randomness of $\{q^{(i)}\}$, the set of  constraint in (\ref{rob_rand}) may be very close to its counterpart (\ref{rob_const}) in the sense of obtaining a small {\em violation probability}, which is now formally defined.

\begin{defi}[Violation probability] Given a decision vector $\theta \in\bR^n$, the violation probability $V(\theta)$ is defined as
\bee
V(\theta):=\bP_{q}\{q\in\cQ |f(\theta,q)>0\}. \nonumber
\ene
\end{defi}
The multi-sample $q^{1:N_{bin}}:=\{q^{(1)},\ldots,q^{(N_{bin})}\}$ is called a {\em scenario} and the resulting optimization problem under the constraint (\ref{rob_rand}) is referred to as a {\em scenario problem} (SP)
\beq
\min_{\theta\in\Theta} c'\theta~~\text{subject to}~~ f(\theta,q^{(i)})\le 0, i=1,\ldots, N_{bin}.\label{rcp_rand}
\enq

In the sequel, let $\Theta^*$ be the set of optimal solutions to the SP and $\Theta_0$ be the set of feasible solutions, i.e.,
\bee
\Theta_0=\{\theta\in\Theta|f(\theta,q^{(i)})\le 0, i=1,\ldots,N_{bin}\}.\label{set_feasible}
\ene

For the SP, we need the following  assumption to study its probabilistic relationship with RCO in (\ref{rcp_opt}).
\begin{assum}[Non-empty set of optimal solutions and interior point] \label{ass_interior}
The SP in (\ref{rcp_rand}) has a non-empty set of optimal solutions, i.e., $\Theta^*\neq \emptyset$. In addition, there exists a vector $\theta_0\in\Theta$ such that
\bee
f(\theta_0,q^{(i)})<0, \forall i=1,\ldots,N_{bin}.\label{interior}
\ene
\end{assum}
The interiority condition (often called Slater's constraint qualification) in (\ref{interior}) implies that there is no duality gap between the primal and dual problems of  (\ref{rcp_rand}) and the dual problem contains at least an optimal solution \cite{bertsekas1999nonlinear}.  We remark that in robust control it is common to study strict inequalities
\cite{petersentempo}, e.g., when dealing with robust asymptotic stability of a system and therefore this is not a serious restriction. In fact, the set of feasible solutions to (\ref{rcp_opt}) is a subset of that of the SP in (\ref{rcp_rand}). The main result of the scenario approach for RCO  is stated below.
\begin{lem}[\cite{campi2008exact}]Assume that there exists a unique solution to (\ref{rcp_rand}). Let $\epsilon$, $\delta \in (0,1)$, and $N_{bin}$ satisfy the following inequality
\bee
\sum_{i=0}^{n-1}\binom{N_{bin}}{i}\epsilon^i(1-\epsilon)^{N_{bin}-i} \le \delta.
\label{scenariosol}
\ene
Then, with probability at least $1-\delta$, the solution $\theta_{sc}$ of the scenario optimization problem (\ref{rcp_rand})
satisfies $V(\theta_{sc}) \le \epsilon$, i.e.,
\bee
\bP_{q^{1:N_{bin}}}\{V(\theta_{sc})\le\epsilon \}\ge 1-\delta.\nonumber
\ene
\end{lem} The uniqueness condition can be relaxed in most
cases by introducing a tie-breaking rule, see Section 4.1 of \cite{calcammp}. If the sample complexity $N_{bin}$ satisfies (\ref{scenariosol}), a solution $\theta_{sc}$ to (\ref{rcp_rand})  approximately solves RCO in (\ref{rcp_opt}) with certain probabilistic guarantee. A subsequent problem is to compute the sample complexity, which dictates the smallest number of constraints required in the SP to solve (\ref{rcp_rand}). This problem has been addressed in \cite{alamo2015randomized} obtaining an improved bound
\bee
N_{bin}\ge \frac{e}{\epsilon(e-1)}(-\ln {\delta}+n-1)\label{scenariosize}
\ene
where $e$ is the Euler's number.
Thus, RCO in (\ref{rcp_opt}) can be approximately solved via the SP in (\ref{rcp_rand}) with a sufficiently large $N_{bin}$.

The remaining objective of this paper is to effectively solve the SP in (\ref{rcp_rand}) when $N_{bin}$ is large.

\section{Distributed Computation Scheme for Scenario Problems}
\label{sec_nps}
In this section, we introduce a distributed computational framework where many processors (nodes) with limited computational capability are interconnected via a  graph. Then, we reformulate the SP in (\ref{rcp_rand}) as a distributed optimization problem, which assigns some local constraints to each node and adapts the coupled constraint to the graph structure.

\subsection{Distributed Computing Nodes}
Although RCO in (\ref{rcp_opt}) can be effectively attacked via the scenario approach, clearly $N_{bin}$ may be large to achieve a high confidence level with small violation probability. For example, in a problem with $n=32$ variables, setting probability levels $\epsilon=0.001$ and $\delta=10^{-6}$, it follows from (\ref{scenariosize}) that the number of constraints in the SP is  $N_{bin}\ge 70898$. For such a large sample complexity $N_{bin}$, the computational cost for solving the SP (\ref{rcp_rand}) becomes very high, which may be far from the computational and memory capacity of a single processor.

To overcome this issue, we propose to use $m$ computing units (nodes) which cooperatively solve the SP in (\ref{rcp_rand})   in a distributed fashion. Then, the number of design constraints for node $j$ is reduced to $n_j$. To maintain the desired probabilistic guarantee, it follows from (\ref{scenariosize}) that
$\sum_{j=1}^m n_j \ge N_{bin}.$

A simple heuristic approach is to assign the number of constraints in (\ref{rcp_rand}) among nodes proportional to their computing and memory power. In practice, each node can declare the total number of constraints that can be handled.  If the number of nodes is comparable to the scenario size $N_{bin}$, the number of constraints for every node $j$ is significantly reduced, e.g. $n_j\ll N_{bin}$, and $n_j$ can be even as small as one.

The problem is then how to  distribute  the computational task  across multiple nodes to cooperatively solve the SP. To this end, we introduce a directed graph $\cG=\{\cV,\cE\}$ to model interactions between the computing nodes  where $\cV:=\{1,\ldots,m\}$ denotes the set of nodes, and the set of links between nodes is represented by $\cE$. A directed edge $(i,j)\in\cE$ exists in the graph if node $i$ directly receives information from node $j$. Then, the in-neighbors and out-neighbors of node $j$ are respectively defined by $\cN_j^{in}=\{i|(j,i)\in\cE\}$ and $\cN_j^{out}=\{i|(i,j)\in\cE\}$. Clearly, every node can directly receive information from its in-neighbors and broadcast information to its out-neighbors.  A sequence of directed edges $(i_1,i_2),\ldots,(i_{k-1},i_k)$ with $(i_{j-1},i_j)\in\cE$ for all $j\in\{2,\ldots,k\}$ is called a directed path from node $i_k$ to node $i_1$. A graph $\cG$ is said to contain a {\em spanning tree} if it has a root node that is connected to any other node in the graph via a directed path, and is  {\em strongly connected} if each node is connected to every other node in the graph via a directed path.

We say that $A=\{a_{ij}\}\in\bR^{m\times m}$ is a {\em row-stochastic} weighting matrix adapted to the underlying graph $\cG$, e.g., $a_{ij}>0$ if $(i,j)\in\cE$ and $0$, otherwise, and $a_{jj}=1-\sum_{i=1, i\neq j}^m a_{ji}\ge 0$ for all $j\in\cV$. Moreover, we denote the associated Laplacian matrix of $\cG$ by $\cL=I_m-A$. If $\cG$ is undirected, $A$ is a symmetric matrix and $\cN_j^{in}=\cN_j^{out}$, which is simply denoted as $\cN_j$.

Overall, the objective of this paper is to solve the following networked optimization problem.
\begin{prob}[Distributed scheme]\label{prob} Assume that $\cG$ is strongly connected. Then, each node  computes a solution to the SP  in (\ref{rcp_rand}) under the following setup:
\begin{enumerate}
\renewcommand{\labelenumi}{\rm(\alph{enumi})}
\item Every node $j$ is able to independently generate $n_j$ i.i.d. samples with an absolutely continuous distribution $\bP_{q}$, and is not allowed to share these samples with other nodes.
\item Every node is able to transmit finite dimensional data per packet via a directed/undirected edge.
\item The vector $c$ in the objective function, the constraint function $f(\theta,q)$ and the set $\Theta$ are accessible to every node.
\end{enumerate}
\end{prob}

In contrast with \cite{carlone2014distributed}, our approach transmits a fixed dimension state vector among nodes. In addition, each node $j$ only deals with a fixed number $n_j$ of constraints. In  \cite{carlone2014distributed}, each node requires to completely solve {\em local} SPs under an increasing number of constraints. We provide a more detailed comparison between our approach and \cite{carlone2014distributed} in Section \ref{rem_comparsion}.

\subsection{Reformulation of the Scenario Problem}

In this work, we propose recursive algorithms with small computation per iteration to distributedly solve the SP. This is particularly suited when several processors cooperate. The main idea is to introduce ``local copies" of $\theta$ in each node, and to optimize and update these variables by incrementally learning the constraints
until a consensus is reached among all the neighboring nodes. The interactions between nodes are made to (indirectly) obtain the constraint set information from other nodes.

Let $q^{(j1)},\ldots,q^{(jn_j)}$ be the samples that are independently generated in node $j$ according to the distribution $\bP_q$. For simplicity, the local constraint functions are collectively rewritten in a vector form
$$
f_j(\theta):=\begin{bmatrix}f(\theta,q^{(j1)}) \\ \vdots \\ f(\theta,q^{(jn_j)}) \end{bmatrix}\in\bR^{n_j}.
$$

Then, the SP in (\ref{rcp_rand}) is equivalent to the following constrained minimization problem
\bee
\min_{\theta\in\Theta} c' \theta~\text{subject to}~f_j(\theta)\preceq 0, \forall j\in\cV, \label{scenarioproblem1}
\ene
where $f_j(\theta)$ is only known to node $j$.

 \begin{exmp}[Continuation of Example \ref{exmp_dro}] \label{rem_dro}
In (\ref{opt_scenarioex1}), the $j$-th component function of $f$ is only known to node $j$. Then, node $j$ can independently extract random samples $\{q_j^{(1)},\ldots,q_j^{(n_j)}\}$ from $\cQ_j$ and obtain the  local inequality
\bee
\tilde{f}_j(\theta,t):=\begin{bmatrix}f_j(\theta,q_j^{(1)})-t_j \\ \vdots \\ f_j(\theta,q_j^{(n_j)})-t_j \end{bmatrix}\preceq 0,
\ene
which is only known to node $j$. Thus, the SP associated with the distributed robust optimization in (\ref{opt_scenarioex1}) has the same form of (\ref{scenarioproblem1}) and can be solved as well.
\end{exmp}

Since each node may have very limited computational and memory capability, the algorithm for each node should be easy to implement with a low computational cost. To achieve this goal, we adopt two different approaches in the sequel for undirected and directed graphs, respectively. The first approach (for undirected graphs) exploits the simple structure of a {\em primal-dual} sub-gradient algorithm \cite{bertsekas1999nonlinear} which has an explicit recursive form. Moreover, the interpretation of this approach is natural from the viewpoint of optimization theory. It requires a bidirectional information flow between nodes and therefore it is not applicable to directed graphs. To overcome this limitation, the second approach (for directed graphs) revisits the idea of Polyak random algorithm for convex feasibility problem \cite{polyak2001random}. We remark that in
\cite{polyak2001random}
the algorithms are centralized and do not address distributed computation, which is resolved in this paper by exploiting the network structure.

Next, we show that the SP can be partially separated by adapting it to the network $\cG$.
\begin{lem}[Optimization equivalence]
\label{lem_equivalence}
Assume that $\cG$ contains a spanning tree. Then, the optimal solution to the SP in (\ref{rcp_rand}) can be found via the following optimization problem
\beq
\min_{\theta_1,\ldots,\theta_m\in\Theta}&&\hspace{-0.5cm} \sum_{j=1}^m c'\theta_j~\text{subject to}~\nonumber \\
&& \hspace{-1.5cm} \sum_{i=1}^ma_{ji}(\theta_j-\theta_i)=0, \label{const_consensus}\\
&&\hspace{-1.5cm} f_j(\theta_j)\preceq 0, \forall j\in\cV.\label{const_ineq}
\enq
\end{lem}
\begin{proof} By a slight abuse of notation, let  $\theta$ be the augmented state of $\theta_j$, i.e., $\theta=[\theta_1',\ldots,\theta_m']'$, and $\cL=I_m-A$, which is the associated Laplician matrix of the graph $\cG$. Then, the constraint in (\ref{const_consensus}) is compactly written as
$
(\cL\otimes I_n){\theta}=0.
$
This is equivalent to $\theta_1=\theta_2=\ldots=\theta_m$ as $\cG$ contains a spanning tree \cite{you2011network}. Thus, the above optimization problem is reduced to 
$$
\min_{\{\theta\in\Theta| f_j(\theta)\preceq 0, \forall j\in\cV\}} (m\cdot c'\theta)
$$
whose set of optimal solutions is equivalent to that of (\ref{rcp_rand}).
\end{proof}

A nice feature of Lemma \ref{lem_equivalence} is that both the objective function and the constraint in (\ref{const_ineq}) of each node are completely decoupled. The only coupled constraint lies in the consensus constraint in (\ref{const_consensus}), which is required to align the state of each node, and can be handled by exploring the graph structure under {\em local} interactions.  Since each node uses it to learn information from every other node, we need the following assumption.
\begin{assum}[Strong connectivity] \label{ass_strong}
The graph $\cG$ is strongly connected.
\end{assum}
As the constraint in (\ref{const_ineq}) is only known to node $j$, this assumption is clearly necessary. Otherwise, there exists a node $i$ that can never be accessed by some other node $j$. In this case, it is impossible for node $j$ to find a solution to the SP (\ref{rcp_rand})  since the information on $f_i(\theta)$ is always missing to node $j$.

\section{Distributed Primal-dual Sub-gradient Algorithms for Undirected Graphs}
\label{sec_npsa}
Recently, several papers concentrated on the distributed optimization problem of the form in Lemma \ref{lem_equivalence}, see e.g. \cite{chatzipanagiotis2015augmented,lee2016asynchronous,nedich2015convergence,varagnolo2016newton,margellos2016distributed,lou2014approximate}  and references therein. However, they mostly consider a generic local constraint set, i.e., the local constraint (\ref{const_ineq}) is replaced by $\theta_j\in\Theta_j$ for some convex set $\Theta_j$, rather than having an explicit inequality form. Thus, the proposed algorithms require a projection onto the set $\Theta_j$ at each iteration to ensure feasibility. This is easy to perform only if $\Theta_j$ has a relatively simple structure, e.g., a half-space or a polyhedron. Unfortunately, the computation of  the projection onto the set
\bee
\Theta_j=\{\theta\in\bR^n|f_j(\theta)\preceq 0\} \label{localset}
\ene is typically difficult and computational demanding. This work does not use projection to handle the inequality constraints. Rather, we exploit the inequality functions by designing distributed primal-dual algorithms for undirected graphs with the aid of an Lagrangian function. Then, we prove that the recursive algorithm in each node asymptotically converges to some common optimal solution of (\ref{rcp_rand}).

Since $\Theta$ is closed and convex, the optimization problem in Lemma \ref{lem_equivalence} is reformulated with equality constraints
\beq
\min&& \hspace{-0.3cm} \sum_{j=1}^m c'\theta_j+h_\rho(\theta)\label{augopt}\\
&&\hspace{-1.5cm}  ~\text{subject to}~ (\cL_j\otimes I_n){\theta}=0,  g_j(\theta_j)=0, \forall j\in\cV\nonumber
\enq
where $\cL_j$ is the $j$-th row of the Laplacian matrix $\cL$, and $g_j(\theta_j)$ is a function only related to the local constraint of node $j$, i.e.,
$$g_j(\theta_j)=\begin{bmatrix} d(\theta_j,\Theta) \\  f_j(\theta_j)_{+} \end{bmatrix}.$$ The distance function $d(\theta,\Theta)$ measures the distance from the point $\theta$ to the set $\Theta$ and is obviously convex in $\theta_j$. Since $\Theta$ is closed and convex, then $d(\theta,\Theta)=0$ if and only if $\theta\in\Theta$.

With a slight abuse of notation, we use $\theta=[\theta_1',\ldots,\theta_m']'$ to denote the augmented state of $\theta_j$. The added quadratic penalty function is defined as
$$
h_\rho(\theta)=\frac{\rho}{2}\sum_{j=1}^m\twon{(\cL_j\otimes I_n)\theta}^2+\|g_j(\theta_j)\|^2
$$
and  $\rho>0$ is a given weighting parameter.

\subsection{Distributed Primal-dual Sub-gradient Algorithm}
To solve the optimization problem (\ref{augopt}), we focus on the following  {\em Lagrangian}
\bee\label{fun_lag}
 L({\theta},\lambda,\gamma)=\sum_{j=1}^m  L_j({\theta},\lambda_j,\gamma_j) \ene
with the {\em local} Lagrangian $L_j({\theta},\lambda_j,\gamma_j)$ defined as
 $$
L_j=c'\theta_j+\lambda_j'(\cL_j\otimes I_n){\theta}+\gamma_j' g_j(\theta_j)+h_\rho(\theta)
$$
where $\lambda_j$ and $\gamma_j$ are the Lagrange multipliers corresponding to (\ref{const_consensus}) and (\ref{const_ineq}), respectively. Then, our objective reduces to find a saddle point $({\theta}^*,\lambda^*,\gamma^*)$ of the Lagrangian $L$ in (\ref{fun_lag}), i.e., for any $({\theta},\lambda,\gamma)$, it holds that
\bee
L({\theta}^*,\lambda,\gamma)\le L({\theta}^*,\lambda^*,\gamma^*)\le L(\theta,\lambda^*,\gamma^*).\label{saddle}
\ene

The existence of a saddle point is ensured under Assumptions \ref{ass_interior} and \ref{ass_strong}, as stated below.
\begin{lem}[Saddle point] \label{lem_saddle}Under Assumptions \ref{ass_interior} and \ref{ass_strong}, there exists a saddle point $({\theta}^*,\lambda^*,\gamma^*)$ of the  Lagrangian $L$ in (\ref{fun_lag}).
\end{lem}
\begin{proof} Under Assumption \ref{ass_interior},  it follows from Propositions 5.1.6 and 5.3.1 in \cite{bertsekas1999nonlinear} that there exists a saddle point for the optimization (\ref{rcp_rand}).  By the equivalence of the SP in (\ref{rcp_rand}) and the problem in Lemma \ref{lem_equivalence}, the rest of proof follows.
\end{proof}
 By the Saddle Point Theorem (see e.g. Proposition 5.1.6 in \cite{bertsekas1999nonlinear}), it is sufficient to find a saddle point of the form (\ref{saddle}). In the section, we design a distributed primal-dual sub-gradient method to achieve this goal.

If $0\preceq\gamma$, then $L({\theta},\lambda,\gamma)$ is convex in each argument, e.g. $L(\cdot,\lambda,\gamma)$ is convex for any fixed $(\lambda,\gamma)$ satisfying $0\preceq\gamma$. Thus, the following set-valued mappings
\beq
T_j({\theta},\lambda,\gamma)&=&\partial_{\theta_j} L({\theta},\lambda,\gamma), \nonumber \\
P_j({\theta},\lambda,\gamma)&=&-\partial_{(\lambda_j,\gamma_j)} L({\theta},\lambda,\gamma)\nonumber
\enq
are well-defined where $\partial_{\theta_j} L({\theta},\lambda,\gamma)$ is the subdifferential of $L$ in $\theta_j$ \cite{bertsekas1999nonlinear}. The optimality of a saddle point $({\theta}^*,\lambda^*,\gamma^*)$ becomes
$0\in T_j({\theta}^*,\lambda^*,\gamma^*)~\text{and}~0\in P_j({\theta}^*,\lambda^*,\gamma^*),$
which is solved via the following iteration
\bee
\theta_j^{k+1}=\theta_j^k-\zeta^k \cdot T_j^k~\text{and}~\nu_j^{k+1}=\nu_j^k-\zeta^k \cdot P_j^k.\label{subgradient}
\ene
Here it is sufficient to arbitrarily select $T_j^k\in T_j(\theta^k,\lambda^k,\gamma^k)$ and $P_j^k\in P_j(\theta^k,\lambda^k,\gamma^k)$. The purpose of $\nu_j^k$ is to compute the Lagrange multipliers of $(\lambda_j^*,\gamma_j^*)$. The stepsizes  satisfy the following condition
\bee\label{stepsize}
\zeta^k>0, ~~ \sum_{k=0}^{\infty}\zeta^k=\infty, ~\text{and}~ \sum_{k=0}^{\infty}(\zeta^k)^2<\infty.
\ene

Next, we show that the sub-gradient iteration in (\ref{subgradient}) can be distributedly computed via Algorithm \ref{alg_dist} for undirected graphs.  For notational simplicity, the dependence of the superscript $k$, which denotes the number of iterations, is removed.  In Algorithm \ref{alg_dist}, every node keeps updating a triple of state vector and Lagrange multipliers $(\theta_j,\lambda_j,\gamma_j)$ by receiving information only from its neighboring nodes $i\in\cN_j$, see Fig. \ref{fig_infor}.  Notice from  (\ref{fun_lag}) that $(\lambda_j,\gamma_j)$ is a pair of Lagrange multipliers that only appears in the local Lagrangian $L_j$. This implies  that
\beq
P_j^k&=&-\begin{bmatrix} \sum\limits_{i=1}^m a_{ji}(\theta_j^k-\theta_i^k)\\
g_j(\theta_j^k) \end{bmatrix}.\nonumber
\enq

Clearly, $\sum_{i=1}^m a_{ji}(\theta_j^k-\theta_i^k)$ in $P_j^k$ is computable in node $j$ by receiving information only from in-neighbors of node $j$.  As $g_j(\theta_j^k)$
is a  function of local variables, $P_j^k$ is accessible to node $j$ via only local interactions with its in-neighbors.
By the additive property of the subdifferential \cite{bertsekas1999nonlinear}, we further obtain  from  (\ref{fun_lag}) that
 \beq
T_j(\theta^k,\lambda^k,\gamma^k)
&=&c+\sum_{i=1}^m l_{ij}\left(\lambda_i^k+\rho \cdot (\cL_i\otimes I_n)\theta^k\right) \nonumber \\
&& +s_j'(\gamma_j^k+\rho\cdot g_j(\theta_j^k)),\nonumber
 \enq
where $l_{ij}$ is the $(i,j)$-th element of the Laplacian matrix $\cL$ and $s_j$ represents a subgradient of $g_j(\cdot)$ at $\theta_j$, i.e., let $\nabla_j$ be a subgradient of $f(\cdot)_+$ at $\theta_j$, then
\bee\label{subgrs}
s_j'=\left [ \frac{\theta_j-\Pi_{\Theta}(\theta_j)}{\twon{\theta_j-\Pi_{\Theta}(\theta_j)}}, \nabla_j'\right]\in\bR^{n\times (n_j+1)}.\ene

Similarly, the second term in the sum $$(\cL_i\otimes I_n)\theta^k= \sum\limits_{j=1}^m a_{ij}(\theta_i^k-\theta_j^k)$$ is locally computable in node $i$.  Together with the fact $\cG$ is undirected, both  in-neighbors and out-neighbors of node $j$ are of the same. Thus, the second term in $T_j(\theta^k,\lambda^k,\gamma^k)$ is obtained by aggregating the modified Lagrange multiplier $\tilde{\lambda}_i^k:=\lambda_i^k+\rho \cdot (\cL_i\otimes I_n)\theta^k$ from its out-neighbors. This further implies that node $j$ is able to compute $T_j(\theta^k,\lambda^k,\gamma^k)$ via local interactions as well. 
\begin{algorithm}[t!]
\caption{Distributed primal-dual algorithm for the SP with undirected graphs}
\label{alg_dist}
\begin{enumerate}
\renewcommand{\labelenumi}{\theenumi:}
\item {\bf Initialization:} Each node $j\in\cV$ sets
$\theta_j=0$, $\gamma_j=0$, and $\lambda_j=0$.
\item {\bf Repeat}
\item {\bf Local information exchange:} Every node $i\in\cV$ broadcasts $\theta_i$ to its neighbor $j\in\cN_i$, computes
$b_i=\sum_{j\in\cN_i}a_{ij}(\theta_i-\theta_j)$ after receiving $\theta_j$ from neighbor $j\in\cN_i$, then broadcasts $\tilde{\lambda}_i=\lambda_i+\rho b_i$ to node $j\in\cN_i$, see Fig. \ref{fig_infor} for an illustration.
\item {\bf Local variables update}: Every node $j\in\cV$ updates $(\theta_j, \lambda_j,\gamma_j)$ as follows
\beq
\lambda_j&\hspace{-0.3cm} \leftarrow&\hspace{-0.3cm} \lambda_j+\zeta\cdot b_j,\nonumber\\
\gamma_j&\hspace{-0.3cm} \leftarrow&\hspace{-0.3cm} \gamma_j+\zeta \cdot g_j(\theta_j),\nonumber\\
\theta_j&\hspace{-0.3cm} \leftarrow&\hspace{-0.3cm}  \theta_j -\zeta \cdot \big(c+s_j'\tilde{\gamma}_j+\sum_{i\in\cN_j}a_{ij}(\tilde{\lambda}_j-\tilde{\lambda}_i)\big)\nonumber
\enq
where $\tilde{\gamma}_j=\gamma_j+\rho \cdot g_j(\theta_j)$, and $s_j$ is a subgradient of $g_j(\cdot)$ at $\theta_j$, see (\ref{subgrs}).

\item {\bf Set} $k=k+1$.
\item {\bf Until} a predefined stopping rule (e.g., a maximum iteration number) is satisfied.
 \end{enumerate}
\end{algorithm}
\begin{figure}[t!]
  \centering
 \includegraphics[width=7cm]{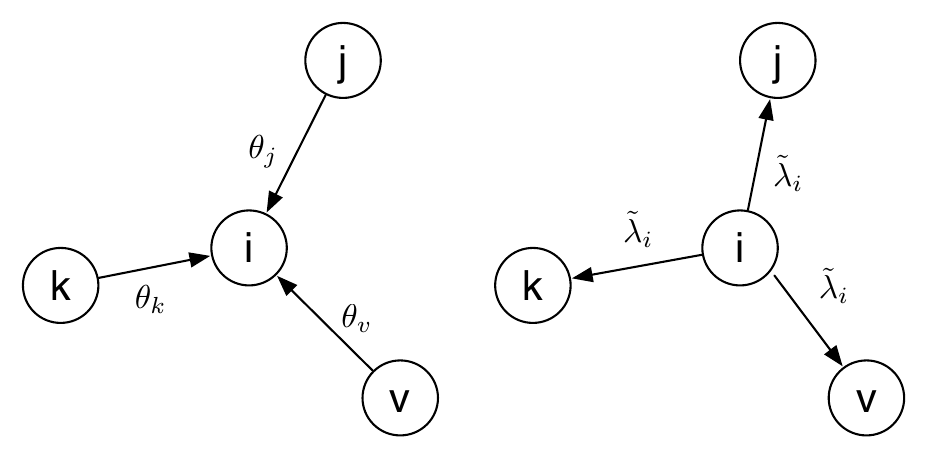}
  \caption{Local information exchange: every node $i$ receives $\theta_j, \forall j\in\cN_i$ from its in-neighbors to compute $b_i$ and $\tilde{\lambda}_i$, after which it broadcasts $\tilde{\lambda}_i$ to out-neighbors.}\label{fig_infor}
\end{figure}

\subsection{Convergence of Algorithm \ref{alg_dist}}

The update of Lagrange multipliers in Algorithm \ref{alg_dist} has interesting interpretation. If $\theta_j$ does not satisfy the local constraint, i.e., ${\bf 1}'g_j(\theta_j)> 0$, some element of the multiplier vector $\gamma_j$ is strictly increased and a larger penalty is imposed on the augmented Lagrangian $L$.  This forces the update of $\theta_j$ to move toward the local feasible set $\Theta\cap \Theta_j$, where $\Theta_j$ is given in (\ref{localset}). If $\gamma_j$ is bounded and the sequence $\{\theta^k\}$ is convergent, it follows that
$\sum_{k=1}^\infty \zeta^k g_j(\theta_j^k) <\infty $ and $\sup_k \twon{P_j^k} <\infty$. In light of (\ref{stepsize}), this implies that $\liminf_{k\rainfty}g_j(\theta_j^k)=0$. Then, the sequence $\{\theta_j^{k}\}$ will eventually enter the local constraint set $\Theta\cap \Theta_j$. Similarly, the multiplier $\lambda^k$ will finally drive the state vector $\theta_j^k$ to reach a consensus in each node. Based on these two observations, it follows that  $\{\theta_j^{k}\}$ finally becomes feasible.  The convergence of Algorithm \ref{alg_dist} is   stated and  proved.

\begin{thm}[Convergence] \label{thm_convergence} Suppose that Assumptions \ref{ass_interior}-\ref{ass_strong} hold and there is a positive $r$ such that $\max\{\twon{T_j^k},\twon{P_j^k}\}\le r$ for all $k$. Then, the sequence  $\{\theta_j^k\}$ of Algorithm \ref{alg_dist} with  stepsizes given in (\ref{stepsize})
converges to some common point in the set $\Theta^*$ of the optimal solutions to (\ref{rcp_rand}).
\end{thm}
\begin{proof}
Let $(\theta^*,\lambda^*, \gamma^*)$ be an arbitrary   saddle point in Lemma \ref{lem_saddle}. Then, it follows from (\ref{subgradient}) that
\beq
&&\hspace{-.5cm}\twon{\theta_j^{k+1}-\theta_j^*}^2= \twon{\theta_j^{k}-\theta_j^*}^2+r^2(\zeta^k)^2 -2\zeta^k(\theta_j^k-\theta_j^*)'T_j^k\nonumber
\enq
Similarly, one can easily obtain $$\twon{\nu_j^{k+1}-\nu_j^*}^2\le \twon{\nu_j^{k}-\nu_j^*}^2+r^2(\zeta^k)^2-2\zeta^k (\nu_j^k-\nu_j^*)'P_j^k.$$ For notational simplicity, let
\bee
z^k=\begin{bmatrix}\theta^k \\ \nu^k\end{bmatrix}, z^*=\begin{bmatrix} \theta^* \\ \nu^*\end{bmatrix}, ~\text{and}~ w^k=\begin{bmatrix}T^k \\ P^k\end{bmatrix}. \nonumber
\ene

Then, summing all $j\in\cV$ leads to that
\beq\label{iteration}
\twon{z^{k+1}-z^*}^2&\le &\twon{z^{k}-z^*}^2+2r^2(\zeta^k)^2\nonumber\\
&& -2 \zeta^k (z^k-z^*)'w^k.
\enq

The rest of the proof is completed by establishing the following two claims.

{\em Claim 1:}  $(z^k-z^*)'w^k \ge 0$ for all $k\ge 1$.

To show the non-negativeness, we write
\beq
&&\hspace{-.7cm}(z^k-z^*)'w^k=\sum_{j=1}^m\big((c+\sum_{i=1}^m l_{ij}\tilde{\lambda}_i^k+s_j'\tilde{\gamma}_j^k )'(\theta_j^k-\theta_j^*)\big)\nonumber\\
&&-(b_j^k)'(\lambda_j^k-\lambda_j^*)-(g_j(\theta_j^k))'(\gamma_j^k-\gamma_j^*)\big), \label{nonnegativeness}
\enq
where $\tilde{\gamma}_j^k=\gamma_j^k+\rho g_j(\theta_j^k)$ is a modified Lagrange multiplier.

Noting that $g_j(\theta_j^*)=0$ and $b_i^*=0$, the sum in (\ref{nonnegativeness}) is split into four sums. The first sum is the difference between two non-penalized  Lagrangians, i.e.,
\bee
\sum_{j=1}^m\left(c'\theta_j^k+(\lambda_j^*)'b_j^k+(\gamma_j^*)'g_j(\theta_j^k) -c'\theta_j^* \right).\nonumber
\ene
The second sum involves the Lagrange multiplier $\lambda^k$, i.e.,
\beq
&&\sum_{j=1}^m\left(\sum_{i=1}^m l_{ij}(\lambda_i^k)'(\theta_j^k-\theta_j^*)-(\lambda_j^k)'b_j^k\right)\nonumber\\
&&=\sum_{i=1}^m(\lambda_i^k)'\sum_{j=1}^m l_{ij}(\theta_j^k-\theta_j^*)-\sum_{j=1}^m(\lambda_j^k)'b_j^k\nonumber\\
&&=\sum_{i=1}^m(\lambda_i^k)'(b_i^k-b_i^*)-\sum_{j=1}^m(\lambda_j^k)'b_j^k\nonumber\\
&&=0 \nonumber
\enq
where we have used the fact that $b_i^*=0$ for all $i\in\cV$.
The third sum involves the Lagrange multiplier $\gamma^k$, i.e.,
\beq
&&\sum_{j=1}^m\big(\gamma_j^k\big)' \left(s_j^k(\theta_j^k-\theta_j^*)-g_j(\theta_j^k)\right) \nonumber \\
&&\ge\sum_{j=1}^m\big(\gamma_j^k\big)'(g_j(\theta_j^k)-g_j(\theta_j^*)-g_j(\theta_j^k))\nonumber\\
&&=0 \nonumber
\enq
where the inequality follows from the fact that $\gamma_j^k\succeq 0$, $g_j(\theta_j^*)=0$ and $s_j^k$ is a sub-gradient of the vector function $g_j(\theta_j)$ at $\theta_j^k$. The fourth sum involves the penalty term, i.e.,
\beq
&&\rho \sum_{j=1}^m \left (\sum_{i=1}^m l_{ij} b_i^k +(s_j^k)'g_j(\theta_j^k) \right)'(\theta_j^k-\theta_j^*) \nonumber\\
&&=\rho\sum_{i=1}^m(b_i^k)'(b_i^k-b_i^*)+g_i(\theta_i^k)'s_i^k(\theta_i^k-\theta_i^*)\nonumber\\
&&\ge \rho \sum_{i=1}^m\big(\twon{b_i^k}^2+\twon{g_i(\theta_i^k)}^2\big)=2h_\rho(\theta^k),\nonumber
\enq
where the inequality follows from $b_i^*=0, g_i(\theta_i^*)=0$ for all $i\in\cV$ and the non-negativeness of $g_i(\theta_i^k)$, together with the fact that $s_i^k$ is a sub-gradient of the vector function $g_i(\theta)$ at $\theta_i^k$. Summing the above four sums, we finally obtain that
\bee
(z^k-z^*)'w^k\ge L(\theta^k,\lambda^*,\gamma^*)- L(\theta^*,\lambda^*,\gamma^*)+h_\rho(\theta^k), \label{nonnegative}
\ene which is non-negative by Lemma \ref{lem_saddle}.

{\em Claim 2:} $\lim_{k\rainfty}\theta_j^k=\lim_{k\rainfty}\theta_i^k\in\Theta^*$ for all $i,j\in\cV$.

To this end, jointly with Proposition A.4.4 in \cite{bertsekas2015convex}, (\ref{stepsize}) and (\ref{iteration}), it follows from  Claim 1  that the sequence $\{\twon{z^k-z^*}\}$ is convergent.  Then, $\twon{z^k}$ is uniformly bounded. This further implies that the subgradient  $\twon{w^k}$ is uniformly bounded, e.g., $\twon{w^k}\le \bar{w}<\infty$ for all $k>0$.
By Claim 1 and Proposition A.4.4 in \cite{bertsekas2015convex}, it follows from (\ref{iteration}) that
\bee
 \sum_{k=1}^\infty\zeta^k (z^k-z^*)'w^k<\infty. \nonumber
\ene
Together with (\ref{stepsize}), we obtain that
\bee
\liminf_{k\rainfty} (z^k-z^*)'w^k=0. \nonumber
\ene

In view of (\ref{nonnegative}), it follows that
$
\liminf_{k\rainfty} L(\theta^k,\lambda^*,\gamma^*)= L(\theta^*,\lambda^*,\gamma^*)$ and $\liminf_{k\rainfty}  h_\rho(\theta^k)=0$. Jointly with (\ref{fun_lag}), we finally obtain that $\liminf_{k\rainfty}\sum_{i=1}^m c'\theta_i^k=\liminf_{k\rainfty}({\bf 1} \otimes c)'\theta^*$ and  $\liminf_{k\rainfty}\theta_i^k=\liminf_{k\rainfty}\theta_j^k$ for all $i,j\in\cV$. That is, there exists an optimal point $\theta_0^*\in\Theta^*$ such that $\liminf_{k\rainfty}\theta_i^k=\theta_0^*$ for all $i\in\cV$.  Moreover, one can easily verify that $({\bf 1}\otimes \theta_0^*, \lambda^*,\gamma^*)$ is also a saddle point of Lemma \ref{lem_saddle}. Together with Claim 1, it holds that $\{\twon{\theta_i^k-\theta_0^*}\}$ converges. Hence, $\lim_{k\rainfty}\theta_i^k=\theta_0^*\in\Theta^*$ for all $i\in\cV$.
\end{proof}
\begin{cor}[Error bounds] Under the conditions of Theorem \ref{thm_convergence},  let
$\check{\theta}^k= {\sum_{t=1}^k\zeta^t\theta^t}/{t^k} ~\text{where}~t^k=\sum_{t=1}^k\zeta^t$. Then, 
\beq
&& L(\check{\theta}^k,\lambda^*,\gamma^*)-L({\theta}^*,\lambda^*,\gamma^*) + h_\rho(\check{\theta}^k) \nonumber\\
&&~~~~\le(\twon{z^1-z^*}^2+2r^2\sum_{t=1}^k(\zeta^t)^2)/(2t^k).
\enq
\begin{proof}It is straightforward by combining (\ref{iteration}) and (\ref{nonnegative}). 
\end{proof}
\end{cor}
\subsection{Comparisons with the State-of-the-art}\label{rem_comparsion} To solve the SP in (\ref{rcp_rand}), a distributed setup is proposed in \cite{carlone2014distributed} by exchanging the active constraints with neighbors. Specifically, each node $j$ solves a {\em local} SP of the form
 \beq
&&\min_{\theta\in\Theta} c'\theta~~\text{subject to}~\nonumber\\
&& \hspace{0.5cm} f(\theta,q^{(i)})\le 0, i\in S_j^k \subseteq \{1,\ldots, N_{bin}\}\label{subrcp_rand}
\enq
at each iteration where $S_j^0$ is the set of indices associated with the random samples generated in node $j$, and obtains local active constraints, indexed as $ActS_j^k:=\{i\in S_j^k|f((\theta_j^k)^*,q^{(i)})=0\}$. Here $(\theta_j^k)^*$ is an optimal solution to the local SP in (\ref{subrcp_rand}), after which it broadcasts its active constraints indexed by $ActS_j^k$ to its out-neighbors. Subsequently, node $j$ updates its local constraint indices as
\bee
S_j^{k+1}=ActS_j^k\cup(\cup_{i\in\cN_j^{in}}ActS_i^k)\cup S_j^{0} \label{constraint}
\ene
and returns a local SP of the form (\ref{subrcp_rand}) replacing $S_j^k$ by $S_j^{k+1}$. In comparison, one can easily identify several key differences from Algorithm \ref{alg_dist}.
\begin{enumerate}
\renewcommand{\labelenumi}{\rm(\alph{enumi})}
\item Using (\ref{subrcp_rand}), we cannot guarantee to reduce the computation cost in each node. In particular, it follows from (\ref{constraint}) that the number of constraints in each local SP in (\ref{subrcp_rand}) increases with respect to the number of iterations, and eventually is greater than the total number of active constraints in the SP in (\ref{rcp_rand}). In an extreme case, the number of active constraints of (\ref{rcp_rand}) can be up to $N_{bin}$. From this point of view, the computation per iteration in each node is still very demanding. It should be noted that selecting the active constraints of an optimization problem is almost as difficult as solving the entire optimization problem.

In Algorithm \ref{alg_dist}, it is clear that the computation only requires a few additions and multiplications of vectors, in addition to finding a sub-gradient of a parameterized function $f(\theta,q)$ in $\theta$. It should be noted that  the computation of the sub-gradient of $f(\theta,q)$ is unavoidable in almost any optimization algorithm. Clearly, the dimension of $\gamma_j$ is  $n_j+1$ and $n_j \approx N_{bin}/m$. This implies that the computation cost in each node is greatly reduced as the number of nodes $m$ increases.

\item Deciding the active constraints in (\ref{subrcp_rand}) is very sensitive to the optimal solution $(\theta_j^k)^*$. If $(\theta_j^k)^*$ is not an exact optimal solution or the evaluation of $f((\theta_j^k)^*,q^{(i)})$ is not exact, we cannot correctly identify the  index set $ActS_j^k$ of active constraints. In Algorithm \ref{alg_dist}, there is no such a problem and the local update has certain robustness properties with respect to the round-off errors in computing $b_j^k$ and $g_j(\theta_j^k)$.

\item The size of data exchange between nodes in (\ref{subrcp_rand}) may grow monotonically.  Although a quantized index version of (\ref{subrcp_rand}) is proposed for the channel with bounded communication bandwidth, it needs to compute the vertices of a convex hull per iteration. More importantly, the dimension of the exchanged data per iteration is still larger than that in Algorithm \ref{alg_dist}.

\item The local SP of the form (\ref{subrcp_rand}) in each node contains several overlapping constraints. Specifically, each constraint  set $\{\theta|f(\theta,q^{(i)})\le 0\}$ could be handled more than once by every node. This certainly induces redundancy in computation. In Algorithm \ref{alg_dist}, each inequality is handled exclusively in only one node. From this perspective, Algorithm \ref{alg_dist} is of great importance for a node with very limited computational and memory capability.
\item It is impossible to describe how the error bounds are reduced with respect to the number of iterations for the distributed algorithms \cite{carlone2014distributed}. 
\end{enumerate}

The primal-dual sub-gradient methods for distributed constrained optimization have been previously used, see e.g., \cite{chang2014distributed}. However, the proposed algorithm originated from the normal Lagrangian (i.e., $\rho=0$ in (\ref{fun_lag})). As discussed in \cite{chang2014distributed} after Theorem 1, this usually requires the strict convexity of the Lagrangian to ensure convergence of the primal-dual sequence, which clearly is not satisfied in our case. To remove this strong convexity condition, the authors propose a specially perturbed sub-gradient and assume boundedness on $\Theta$ and $\partial_\theta f(\theta,q^{(i)})$. This increases the complexity of the distributed algorithm. In particular, it requires to run up to three consensus algorithms and projects the dual variable onto a bounded ball whose radius must be initially decided, and it is a global parameter. Obviously, Algorithm \ref{alg_dist} has a much simpler structure by adopting an augmented Lagrangian in (\ref{fun_lag}), which, to some extent, can be interpreted  as the strict convexification of the Lagrangian function. Moreover, the convergence proof of Algorithm \ref{alg_dist}, which is given in the next subsection,  is simpler and easier to understand.

Compared with the distributed alternating direction method of multipliers (ADMM) \cite{shi2014linear,chang2015multi,iutzeler2016explicit},  the computation of Algorithm \ref{alg_dist} is  simpler. For example, the ADMM essentially updates the primal sequence as follows
\bee
\theta^{k+1}\in\argmin_{\theta\in\Theta^m}  L_c({\theta},\lambda^k,\gamma^k) \label{admm}
\ene
where $L_c({\theta},\lambda^k,\gamma^k)$ has a similar form to the augmented Lagrangian $L({\theta},\lambda^k,\gamma^k)$ in (\ref{fun_lag}).  That is, it requires to solve an optimization (\ref{admm}) per iteration. In Algorithm \ref{alg_dist},  we only need to compute one inner iteration to update $\theta^k$ by moving along the sub-gradient direction.
\subsection{Extensions to Stochastically time-varying graphs}
\label{sec_est}
Algorithms \ref{alg_dist} can be easily generalized to the case of stochastically time-varying graphs with a fixed number of nodes. In particular, let the interaction graph at time $k$ be $\cG^k:=\{\cV,\cE^k\}$. If $\{\cG^k\}$ is an i.i.d. process where the mean graph $\bE[\cG^k]$ is strongly connected, Theorem \ref{thm_convergence}
continues to hold by following similar lines of proof. For instance, it is easy to show that the SP in (\ref{rcp_rand}) is equivalent to
\beq
\min_{\theta_1,\ldots,\theta_m \in\Theta}&& \hspace{-0.3cm} \sum_{j=1}^m c'\theta_j~\text{subject to}~\label{rcp_rande}\\
&& \hspace{-1.5cm} \sum_{i=1}^m\bE[a_{ji}^k](\theta_j-\theta_i)=0,   f(\theta_j,q^{(j)})\preceq 0, \forall j\in\cV.\nonumber
\enq

Next, consider a stochastically time-varying augmented Lagrangian
\beq\label{fun_lag1}
&&  L^k({\theta},\lambda,\gamma)=\sum_{j=1}^mL_j^k(\theta,\lambda_j,\gamma_j),
\enq
where $L_j^k$ is obtained by replacing $\cL_j$ with $\cL_j^k$ in (\ref{fun_lag}). Moreover, all the elements $a_{ij}$ in Algorithm \ref{alg_dist} are replaced by $a_{ij}^k$. Using the  theory of stochastic approximation \cite{kushner2003stochastic}, we can find a saddle point of  $\bE[L^k]$, i.e.,  for any $ ({\theta},\gamma,\lambda)$, the inequalities
\beq
\bE[L({\theta}^*,\lambda,\gamma)] \le\bE[L({\theta}^*,\lambda^*,\gamma^*)] \le \bE[L({\theta},\lambda^*,\gamma^*)]\nonumber
\enq
hold almost surely. Following a similar reasoning, we can establish the following result, the proof of which is omitted due to the page limitation. 
\begin{thm}[Almost sure convergence] Let Assumption \ref{ass_interior} hold and let $\{\cG^k\}$ be an i.i.d. sequence with $\bE[\cG^k]$ strongly connected. If there exists a positive $r$ such that $\max\{\twon{T_j^k},\twon{P_j^k}\}\le r$, the sequence  $\{\theta_j^k\}$ of Algorithm \ref{alg_dist} with stepsizes in (\ref{stepsize}) and $a_{ij}$ replaced by $a_{ij}^k$ converges almost surely to some common random point in the set $\Theta^*$ of the optimal solutions to (\ref{rcp_rand}).
\end{thm}

\section{Distributed Random Projected Algorithms for Directed Graphs}
\label{sec_nrpa}
In this section, we are concerned with the design of a distributed algorithm for directed graphs.  Different from undirected graphs, the information flow between nodes is unidirectional, which results in information unbalance of the network, and renders the primal-dual algorithm inapplicable. To overcome it, we design a consensus algorithm to gather information from in-neighbors and obtain an intermediate state vector. The feasibility is then asymptotically ensured by driving the intermediate state vector toward the local constraint set, which is achieved by updating the solution toward the sub-gradient direction of a randomly selected constraint function.  This process is realized by designing a novel {\em distributed} variation of a Polyak random algorithm \cite{nedic2011random}, see further comments in Remark \ref{remarkpolyak}.
The main result is then to prove almost sure convergence of an optimal solution.

\subsection{Distributed Random Projected Algorithm}
In Fig. \ref{fig_infor}, it is clear that the information exchange is bidirectional. In particular, Algorithm \ref{alg_dist} requires each node $j$ to use the {\em modified} Lagrangian multipliers $\tilde{\lambda}_i$ from its out-neighbors to update the decision vector $\theta_j$. Obviously, this is not implementable for directed graphs, and in this case there is no clear way to design a distributed primal-dual algorithm. For this purpose,  we propose a two-stage  distributed random projected algorithm
\beq
v_j^k&=&\sum_{i=1}^ma_{ji}\theta_i^k-\zeta^k\cdot c,\label{dist_mix}\\
\theta_j^{k+1}&=&\Pi_{\Theta}(v_j^k-\beta \cdot \frac{f(v_j^k,q^{(jw_j^k)})_+}{\twon{d_j^k}^2}d_j^k), \label{dist_digraph}
\enq
where $\zeta^k>0$ is the (deterministic) stepsize given in (\ref{stepsize}) $\beta\in(0,2)$ is a constant parameter, $w_j^k\in\{1,\ldots,n_j\}$ is a random variable and the vector $d_j^k\in\partial f(v_j^k,q^{(jw_j^k)})_+$ if $f(v_j^k,q^{(jw_j^k)})_+>0$ and $d_j^k=d_j$ for some $d_j\neq 0$ if $f(v_j^k,q^{(jw_j^k)})_+=0$.

We intuitively explain the key ideas of the above algorithm. The objective of (\ref{dist_mix}) is to distributedly solve an unconstrained optimization, i.e., the optimization by removing the constraints in (\ref{const_ineq}), see \cite{nedic2009distributed} for details. Note that in
\cite{nedic2009distributed} the double stochasticity of $A$ is required, which is in fact not necessary in our paper. The purpose of (\ref{dist_digraph}) is to drive the intermediate state $v_j^k$ toward a randomly selected local constraint set $\Theta\cap\Theta_j^{w_j^k}$, where $\Theta_j^{w_j^k}:=\{\theta|f(\theta,q^{(jw_j^k)})\le 0\}$. If $\beta$ is sufficiently small, it is easy to verify (see e.g. \cite[Proposition 6.3.1]{bertsekas1999nonlinear}) that $$d(\theta_j^{k+1},\Theta\cap\Theta_j^{w_j^k})\le d(v_j^k,\Theta\cap\Theta_j^{w_j^k}).$$ That is, $\theta_j^{k+1}$ is closer to the local constraint set $\Theta\cap\Theta_j^{w_j^k}$ than $v_j^k$. If $w_j^k$ is uniformly selected at random from $\{1,\ldots,n_j\}$, we conclude that $\theta_j^{k+1}$ is closer to the local constraint set $\Theta\cap\Theta_j$ than $v_j^k$ in the average sense.  Once the consensus is achieved among nodes, the state vector $\theta_j^k$ in each node asymptotically converges to a point in the feasible set $\Theta_0$.

\begin{rem}
\label{remarkpolyak}
The proposed algorithm is motivated by a generalized Polyak random algorithm \cite{nedic2011random}, which however does not address the distributed design. In this paper, we adapt this algorithm to a directed graph with multiple interconnected nodes and establish its asymptotic optimality for strongly connected digraphs. To the best of our knowledge, the existing work on distributed optimization mostly require the underlying graph to be {\em balanced} of the form that the weighting matrix $A$ is doubly stochastic, see e.g. \cite{nedic2009distributed,duchi2012dual,gharesifard2014distributed,lee2016asynchronous}. Clearly, assuming that the directed graph is balanced is a quite restrictive assumption on the network topology, which is in fact not necessary.
This issue has been recently resolved either by combining the gradient descent and the push-sum consensus \cite{nedic2015distributed}, or augmenting an additional variable for each agent to record the state updates  \cite{xi2015directed}. In comparison, the algorithm in \cite{nedic2015distributed} only focuses on the {\em unconstrained} optimization, involves nonlinear iterations and requires the updates of four vectors. The algorithm in \cite{xi2015directed} requires an additional ``surplus" vector to record the state update, which increases the computation and communication cost. From this viewpoint, the proposed algorithm of this paper has a simpler structure and is easier to implement, see Algorithm~\ref{alg_dist1} for details.
\end{rem}
\subsection{Convergence of Algorithm \ref{alg_dist1}}
To prove convergence, we need the following assumptions, most of which are standard in sub-gradient methods.

\begin{algorithm}[t!]
\caption{Distributed random projection algorithm for the SP with directed graphs}
\label{alg_dist1}
\begin{enumerate}
\renewcommand{\labelenumi}{\theenumi:}
\item {\bf Initialization:} For each node $j\in\cV$ set
$\theta_j=0$.
\item {\bf Repeat}
\item {\bf Local information exchange:} Every node $j\in\cV$ broadcasts $\theta_j$ to its out-neighboring nodes.
\item {\bf Local variables update}: Every node $j\in\cV$ receives the state vector $\theta_i$ from its in-neighbor $i\in\cN_j^{in}$ and updates it as follows
\begin{itemize}
\item $v_j =\sum_{i\in\cN_j^{in}}a_{ji}\theta_i-\zeta c$ where the stepsize $\zeta$ is given in (\ref{stepsize}).
\item Draw $w_j\in \{1,\ldots,n_j\}$ uniformly at random.
\item $\theta_j\leftarrow\Pi_{\Theta}(v_j-\beta \cdot \frac{f(v_j,q^{(jw_j)})_+}{\twon{d_j}^2}d_j)$ where $d_j$ is defined in (\ref{dist_digraph}).
\end{itemize}
\item {\bf Set} $k=k+1$.
\item {\bf Until} a predefined stopping rule (e.g., a maximum iteration number) is satisfied.
 \end{enumerate}
\end{algorithm}

\begin{assum}[Randomization and sub-gradient boundedness] \label{ass_random} Let the following hold:
\begin{itemize}
\item[(a)] $\{w_j^k\}$ is an i.i.d. sequence that is uniformly distributed over the set $\{1,\ldots,n_j\}$ for any $j\in\cV$, and is independent over the index $j$.
\item [(b)] The sub-gradients $d_j^k$ are uniformly bounded over the set $\Theta$, i.e., there exists a scalar $r$ such that
$$\twon{d_j^k}\le r, \forall j\in\cV.$$
\end{itemize}
\end{assum}

Clearly, the designer is free to choose any distribution for drawing the samples $w_j^k$. Thus, Assumption \ref{ass_random}(a) is easy to satisfy. By the property of the sub-gradient and (\ref{dist_digraph}), a sufficient condition for Assumption \ref{ass_random}(b) is that $\Theta$ is bounded.  

We now  present the convergence result on the distributed random algorithm.
\begin{thm}[Almost sure convergence]\label{thm_drp}
Suppose that Assumptions \ref{ass_interior}-\ref{ass_random} hold. The sequence $\{\theta_j^k\}$ of Algorithm \ref{alg_dist1} converges almost surely to some common point in the set $\Theta^*$ of the optimal solutions to (\ref{rcp_rand}).
\end{thm}
\subsection{Proof of Theorem \ref{thm_drp}}

The proof is roughly divided into three parts. The first part establishes a stochastically ``decreasing" result, see Lemma \ref{lem_sigmafield}. That is, the distance of $\theta^{k+1}$ to some optimal point $\theta^*$ is ``stochastically"
closer than that of $\theta^k$. The second part essentially shows the asymptotic feasibility of the state vector $\theta_j^k$, see Lemma \ref{lem_thetay}. Finally, the last part establishes an asymptotic consensus  result in Lemma \ref{lem_mixing}, which shows that the sequence $\{\theta_j^k\}$ converge to some common value for all $j\in\cV$. Combining these results, we show that $\{\theta_j^k\}$ converges almost surely to some common random point in the set $\Theta^*$.

Now, we establish a stochastically ``decreasing" result.
\begin{lem}[Stochastically decreasing] \label{lem_sigmafield}
Let $\cF^k$ be the sigma-field generated by the random variables $\{w_j^t, j\in\cV\}$ up to time $k$, i.e.,
\bee\cF^k=\{w^0,\ldots,w^k\}\label{sigmafield}
\ene and $\hat{\theta}_j^{k}=\sum_{i=1}^ma_{ji}\theta_i^k$, where $\theta_j^k$ is generated in  Algorithm \ref{alg_dist1}. 

Under Assumptions \ref{ass_interior} and \ref{ass_random}, it holds almost surely that for all $j\in\cV$ and $k \ge \tilde{k}$,  which is a sufficiently large number,
\beq
&&\hspace{-.8cm}\bE[\twon{\theta_j^{k+1}-\theta^*}^2|\cF_{k}]\le \left(1+r_1(\zeta^k)^2\right)\twon{\hat{\theta}_j^{k}-\theta^*}^2 \label{key_inequality} \\
&&\hspace{.4cm} -2\zeta^kc'(y_j^k-\theta^*) - r_2(\twon{\hat{\theta}_j^{k}-y_j^k}^2)+r_3(\zeta^k)^2,  \nonumber
\enq
where $r_i>0, i\in\{1, 2, 3\}$, $\theta^*\in\Theta^*$ and $y_j^k=\Pi_{\Theta_0}(\hat{\theta}_j^{k})$ with $\Theta_0$ given in (\ref{set_feasible}).
\end{lem}
\begin{proof} The proof mostly follows from \cite{nedic2011random}, which however only focuses on the centralized version of Algorithm \ref{alg_dist1}. By the comments after Assumption 2 of \cite{nedic2011random}, it is clear that all conditions in \cite[Proposition 1]{nedic2011random}  are satisfied. By the row stochasticity of $A$, i.e., $\sum_{i=1}^m a_{ji}=1$, it follows that (\ref{dist_mix}) can be also written as
\bee
v_j^k=\hat{\theta}_j^{k}-\zeta^k\cdot \nabla \left( c'\hat{\theta}_j^k\right),\nonumber
\ene
where $\nabla \left( c'\hat{\theta}_j^k\right)$ is a gradient of the linear function $c'\theta$ evaluated at $\theta=\hat{\theta}_j^{k}$. The rest of proof is trivial by replacing $x_{k-1}$ in (21) of \cite{nedic2011random} with $\hat{\theta}_j^{k}$. The details are omitted.
\end{proof}
The second result essentially ensures the local feasibility.
\begin{lem}[Feasibility guarantee] \label{lem_thetay}Let $y_j^k$ be given in Lemma \ref{lem_sigmafield}. If $\lim_{k\rainfty}\twon{v_j^k-y_j^k}=0$,  it holds
$
\lim_{k\rainfty}\twon{\theta_j^{k+1}-y_j^k}=0
$
for any $j\in\cV$.
\end{lem}
\begin{proof}
Since $f(y_j^k,q^{(jw_j^k)})_+=0$, it follows from Lemma 1 of \cite{nedic2011random} that
$$
\twon{\theta_j^{k+1}-y_j^k}^2 \le \twon{v_j^k-y_j^k}^2-\beta(2-\beta)\frac{\big(f(v_j^k,q^{(jw_j^k)})_+\big)^2}{\twon{d_j^k}^2}.
$$
Together with the fact that $\beta\in(0,2)$, then $\twon{\theta_j^{k+1}-y_j^k}^2 \le \twon{v_j^k-y_j^k}^2$. Taking limits on both sides, the result follows. \end{proof}

Finally, we prove an asymptotic consensus result under Assumption \ref{ass_strong} where the consensus value is a weighted average of the state vector in each node.  This is different than
the case of balanced graphs. For a strongly connected digraph $\cG$, we have some preliminary results on its weighting matrix $A$ by directly using the Perron Theorem \cite{horn1985ma}.
\begin{lem}[Left eigenvector]\label{lem_lefteig}Under Assumption \ref{ass_strong}, there exists a normalized left eigenvector $\pi\in\bR^m$ of $A$ such that
\bee \label{lefteig}
\pi'A=\pi', \sum_{j=1}^m\pi_j=1~\text{and}~\pi_j>0, \forall j\in\cV.
\ene

Moreover, the spectral radius of the row-stochastic matrix $A-{\bf 1}\pi'$ is strictly less than one.
\end{lem}
\begin{lem}[Asymptotic consensus] \label{lem_mixing} Consider the following iteration
\bee
\theta_j^{k+1}=\sum_{i=1}^m a_{ji}\theta_i^k+n_j^k, \forall j\in\cV. \nonumber
\ene

Suppose that $\cG$ is strongly connected and $\lim_{k\rainfty}\twon{n_j^k}=0$.  Let $\bar{\theta}^k=\sum_{i=1}^m\pi_i\theta_i^k$, where $\pi_i$ is given in (\ref{lefteig}), it holds that
\bee
\lim_{k\rainfty}\twon{\theta_j^k-\bar{\theta}^k}=0, \forall j\in\cV. \label{consensus}
\ene
\end{lem}
\begin{proof} Clearly, we can compactly write $\bar{\theta}^k=(\pi'\otimes I_n)\theta^{k}$. In view of  (\ref{lefteig}) and (\ref{consensus}), we have the following relation
\bee
{\bf 1}(\pi'\otimes I_n)\theta^{k+1}={\bf 1}(\pi'\otimes I_n)\theta^{k}+{\bf 1}(\pi'\otimes I_n)n^k. \label{average}
\ene

Let $\delta^k=((I_n-{\bf 1}\pi')\otimes I_n)\theta^k$, which is a vector of displacement from the weighted average. Then, it follows from (\ref{average}) that
\bee
\delta^{k+1}=((A-{\bf 1}\pi')\otimes I_n)\delta^k+((I-{\bf 1}\pi')\otimes I_n)n^k.\nonumber
\ene

Define $\varrho$ as the spectral radius of $(A-{\bf 1}\pi')\otimes I_n$, it is clear from Lemma \ref{lem_lefteig} that $0<\varrho<1$.  Jointly with the fact that $\lim_{k\rainfty}\twon{n^k}=0$ and    Lemma 6.1.1\cite{ash2000pam}, it follows that $\lim_{k\rainfty}\twon{\delta^k}=0$.
\end{proof}

The proof also depends crucially on the well-known super-martingale convergence theorem, which is due to Robbins-Siegmund \cite{robbins1985convergence}, see also Proposition A.4.5 in \cite{bertsekas2015convex}. This result is now restated for completeness.
\begin{thm}[Super-martingale convergence theorem]\label{thm_martingale}
Let $\{y^k\}, \{z^k\}, \{w^k\}$ and $\{v^k\}$ be four non-negative sequences of random variables, and let $\cF^k, k=0,1,\ldots,$ be sets of random variables such that $\cF^k\subseteq\cF^{k+1}$ for all $k$. Assume that
\begin{enumerate}
\renewcommand{\labelenumi}{\rm(\alph{enumi})}
\item For each $k$, let $y^k,z^k,w^k$ and $v^k$ be functions of the random variables in $\cF^k$.
\item The inequalities hold almost surely
$$\bE[y^{k+1}|\cF^k]\le(1+v^k)y^k-z^k+w^k, k=0,1,\ldots, \text{and}$$$$\sum_{k=0}^\infty w^k<\infty,~\sum_{k=0}^\infty v^k<\infty.$$
\end{enumerate}
Then, $\{y^k\}$ converges almost surely to a nonnegative random variable $y$, and $\sum_{k=0}^\infty z^k <\infty$.
\end{thm}

Combine the above, we are ready to prove Theorem  \ref{thm_drp}.

{\em Proof of Theorem \ref{thm_drp}.} By the convexity of $\twon{\cdot}^2$ and the row stochasticity of $A$, i.e., $\sum_{i=1}^m a_{ji}=1$, it follows that
\bee
\twon{\hat{\theta}_j^{k}-\theta^*}^2 \le \sum_{i=1}^m a_{ji}\twon{\theta_i^k-\theta^*}^2.\nonumber
\ene
Jointly with (\ref{key_inequality}), we obtain that for all $k\ge \tilde{k}$,
\begin{eqnarray}
&&\bE[\twon{\theta_j^{k+1}-\theta^*}^2|\cF_{k}]\le \left(1+r_1(\zeta^k)^2\right)\sum_{i=1}^ma_{ji}\twon{{\theta}_i^k-\theta^*}^2 \nonumber\\
&&-2\zeta^kc'(y_j^k-\theta^*) - r_2(\twon{\hat{\theta}_j^{k}-y_j^k}^2)+r_3(\zeta^k)^2,\label{inequal}
\end{eqnarray}
where the sigma-field $\cF^k$ is given in (\ref{sigmafield}).

Under Assumption \ref{ass_strong}, the weighting matrix $A$ of $\cG$ is only row stochastic, and not doubly stochastic, which is assumed in \cite{lee2016asynchronous}. This implies that the first term in (\ref{dist_mix}) does not satisfy average consensus. Instead, it converges to the weighted average consensus where the weight is determined by the left eigenvector $\pi\in\bR^m$ of $A$ associated with the simple eigenvalue $1$, i.e., $\pi' A=\pi$, see Lemma \ref{lem_mixing}. Since the graph $\cG$ is strongly connected, it is clear that $\pi_j>0$ for all $j\in\cV$.

Then,  we multiply both sides of (\ref{inequal}) with $\pi_j$ and sum over $j$, which leads to
\begin{eqnarray}
&&\bE[\sum_{j=1}^m\pi_j \twon{\theta_j^{k+1}-\theta^*}^2|\cF_{k}] \label{inequalityd}\\
&&\le \left(1+r_1(\zeta^k)^2\right)\sum_{j=1}^m \pi_j \big (\sum_{j=1}^m a_{ji} \twon{\theta_j^{k}-\theta^*}^2\big ) \nonumber \\
&&-2\zeta^kc'(\bar{y}^k-\theta^*) - \sum_{j=1}^m \pi_j \left( r_2(\twon{\hat{\theta}_j^{k}-y_j^k}^2)+r_3(\zeta^k)^2\right) \nonumber \\
&&\le \left(1+r_1(\zeta^k)^2\right)\sum_{j=1}^m\pi_j \twon{\theta_j^{k}-\theta^*}^2 \nonumber \\
&&-2\zeta^kc'(\bar{y}^k-\theta^*) -r_2 \sum_{j=1}^m \pi_j (\twon{\hat{\theta}_j^{k}-y_j^k}^2)+r_3(\zeta^k)^2\nonumber
\end{eqnarray}
where the first inequality uses the fact that $\bar{y}^k=\sum_{j=1}^m\pi_j y_j^k$ and $\sum_{j=1}^m \pi_j=1$. The second inequality follows from the definition of $\pi$, i.e., $\pi_j=\sum_{i=1}^m \pi_ia_{ji}$.

By Theorem \ref{thm_martingale}, it holds almost surely that $\{\sum_{j=1}^m\pi_j\twon{\theta_j^{k}-\theta^*}^2\}$ is convergent for any $j\in\cV$ and $\theta^*\in\Theta^*$,
\bee
\sum_{k=1}^{\infty}\zeta^kc'(\bar{y}^k-\theta^*)<\infty \label{martingale1}
\ene
and
\bee
\sum_{k=1}^{\infty}\sum_{j=1}^m \pi_j \twon{\hat{\theta}_j^{k}-y_j^k}^2<\infty.\label{martingale2}
\ene

The rest of the proof is completed by showing the following two claims.

{\em Claim 1:}~$\{\twon{\bar{y}^k-\theta^*}\}$ converges almost surely.

In light of (\ref{martingale2}), it holds that $\{\twon{y_j^k-\hat{\theta}_j^{k}}\}$ converges to zero almost surely. Since $\zeta^k \ra 0$, it follows from (\ref{dist_mix}) that $\{\twon{v_j^k-\hat{\theta}_j^k}\}$ converges almost surely to zero as well. Combing the preceding two relations, it holds almost surely that $\lim_{k\rainfty}\twon{y_j^k-v_j^k}=0$. Together with Lemma \ref{lem_thetay}, it holds almost surely that $
\lim_{k\rainfty}\twon{\theta_j^{k+1}-y_j^k}=0
$
for any $j\in\cV$.  Since $\{\sum_{j=1}^m\pi_j \twon{\theta_j^{k+1}-\theta^*}^2\}$ converges almost surely, this implies that $\{\sum_{j=1}^m\pi_j \twon{y_j^k-\theta^*}^2\}$ converges as well.

By (\ref{dist_mix}) and (\ref{dist_digraph}), we have the following dynamics
\bee
\theta_j^{k+1}=\sum_{i=1}^ma_{ji}\theta_i^k+n_j^k
\ene
where $n_j^k=\theta_j^{k+1}-v_j^k-\zeta^k c$. Since the inequality
\bee
\twon{n_j^k}\le \twon{\theta_j^{k+1}-y_j^k}+\twon{y_j^k-v_j^k}+\zeta^k \twon{c}, \nonumber
\ene
holds, it is obvious that $\lim_{k\rainfty}\twon{n_j^k}=0$ almost surely. Together with Lemma \ref{lem_mixing}, we obtain that $\lim_{k\rainfty}\twon{\theta_j^k-\bar{\theta}^k}=0$ almost surely.

Since $\pi_j=\sum_{i=1}^ma_{ji}\pi_i$, it holds that $\bar{\theta}^k=\sum_{j=1}^m \pi_j \theta_j^k =\sum_{i=1}^m \pi_i \hat{\theta}_i^k$. Then, we obtain that
\beq
\twon{\hat{\theta}^k_j-\sum_{i=1}^m \pi_i \hat{\theta}_i^k}\hspace{-.5cm}&&=\twon{\sum_{i=1}^m a_{ji}{\theta}_i-\bar{\theta}^k}\nonumber\\
&&\hspace{-.5cm}\le \sum_{i=1}^m a_{ji}\twon{\theta_i^k-\bar{\theta}^k}\ra 0~\text{as}~k\rainfty.\nonumber
\enq

Since $\lim_{k\rainfty}\twon{y_j^k-\hat{\theta}_j^{k}}=0$, it follows that
$
\twon{y^k_j-\bar{y}^k}\le \twon{y_j^k-\hat{\theta}_j^{k}}+\twon{\hat{\theta}^k_j-\sum_{i=1}^m \pi_i \hat{\theta}_i^k}+\sum_{i=1}^m \pi_i
\twon{\hat{\theta}_i^k-y_i^k}
$, which converges almost surely to zero as $k\rainfty$ by using the above relations. Jointly with the fact that $\{\sum_{j=1}^m\pi_j \twon{y_j^k-\theta^*}^2\}$ converges, we obtain that $\{\twon{\bar{y}^k-\theta^*}\}$ converges almost surely.

{\em Claim 2:} There exists $\theta^*_0\in\Theta^*$ such that $\lim_{k\rainfty}\theta_j^k=\theta^*_0$ for all $j\in\cV$ with probability one.

By  (\ref{stepsize}) and (\ref{martingale1}), it follows that $\liminf_{k\rainfty}c'\bar{y}^k=c'\theta^*$, which implies that there exists a subsequence of $\{\bar{y}^k\}$ that converges almost surely to some point in the optimal set $\Theta^*$, which is denoted as $\theta^*_0$. Jointly with Claim 1 that $\{\twon{\bar{y}^k-\theta^*_0}\}$ converges, it follows that $\lim_{k\rainfty}\bar{y}^k=\theta^*_0$ almost surely.  Finally, we note that $\twon{\theta_j^{k+1}-\theta^*_0}\le \twon{\theta_j^{k+1}-y_j^{k+1}}+\twon{y_j^{k+1}-\bar{y}^{k+1}}+\twon{\bar{y}^{k+1}-\theta^*_0}$, which converges almost surely to zero as $k\rainfty$. Thus,  Claim 2 is proved.\qed
\begin{cor}[Error bounds] Under the conditions of Theorem \ref{thm_drp}, let $\check{y}^k=\frac{1}{t^k}\sum_{t=1}^k\zeta^t \bar{y}^t$ and $e^k=c'(\check{y}^k-\theta^*)$. Then, for all $k \ge \tilde{k}$, it holds that
\bee
0\le \bE[e^k]\le \frac{c^k}{2t^k} ~\text{and}~ \bE[\twon{\theta_j^k-y_j^k}^2]\le \frac{c^k}{a_{jj}\pi_j k}
\ene
\end{cor}
where $y_j^k=\Pi_{\Theta_0}(\hat{\theta}_j^{k})$ is feasible and $$c^k=\exp(r_1\sum_{t=1}^k(\zeta^t)^2)\big(\sum_{j=1}^m \pi_j \twon{\theta_j^1-\theta^*}^2+r_3\sum_{t=1}^k(\zeta^t)^2\big).$$  
\begin{proof} Note that $y_j^k$ is feasible and $\prod_{t=1}^k(1+r_1(\zeta^t)^2)\le \prod_{t=1}^k \exp(r_1(\zeta^t)^2)<\infty$. By  (\ref{inequalityd}),  the proof requires tedious but easy algebraic operations and is omitted to save space.  
\end{proof}

As in Section \ref{sec_est}, Algorithm \ref{alg_dist1} can also be modified to deal with the case of stochastically time-varying graphs.

\subsection{Comparison with the Distributed Primal-Dual Algorithm}

In this subsection, we compare the previously two algorithms. First, although both algorithms are designed from different perspectives, they essentially converge as fast as $O(1/\sum_{t=1}^k\zeta^t)$.  Let $0<\alpha\le 0.5$, it follows from  (\ref{stepsize}) that it suffices to select $\zeta^t=t^{-(0.5+\alpha)}$, and
$$\sum_{t=1}^k\zeta^t \approx\int_{0}^k t^{-(0.5+\alpha)} dt=\left\{ \begin{array}{ll}k^{0.5-\alpha}, &\text{if}~0<\alpha<0.5, \\ \ln k, &\text{if}~\alpha=0.5. \end{array}\right.$$ This implies that the convergence rate of both algorithms can be as fast as $O(1/\sqrt{k})$, which is an optimal rate for a generic sub-gradient algorithm, see Page 9 in \cite{boyd2011subgradient}.

Second, the primal-dual algorithm is originated from sub-gradient methods for finding a saddle point of the augmented Lagrangian. In \cite{bertsekas2015convex}, there are quite a few methods to accelerate the sub-gradient method, which may provide many opportunities  to accelerate the networked primal-dual algorithms. This is not obvious for Algorithm \ref{alg_dist1} since there is no clear way to accelerate its convergence.

Third, the computational cost of both algorithms is low at each iteration. The algorithms are well-suited for the computing nodes with limited computation and memory capability.

\section{Application Example: Robust Identification}
\label{sec_simulation}
To illustrate effectiveness of the proposed distributed algorithms, we consider a RCO problem in (\ref{rcp_opt}) with linearly structured uncertainties in an identification problem where we seek to estimate the impulse response $\theta$ of a discrete-time system for its input $u$ and output $y$.

Assume that the system is linear, single input single output and of order $n$, and that $u$ is zero for negative time indices and $\theta,u$ and $y$ are related by the convolution equations $y=U\theta$ where $U$ is a lower-triangular Toeplitz matrix whose first column is $u$, i.e., let $u=[u_1,\ldots,u_n]'$, then
$$U=\begin{bmatrix}u_1 &0&\ldots &0 \\ u_2&u_1&\ldots & 0 \\ \vdots & \vdots & \ddots &\vdots \\ u_n & u_{n-1} & u_2&u_1\end{bmatrix}.$$

 Suppose that the actual input and output are $u+\delta u$ and $y+\delta y$, respectively. Then, the standard least squares (LS) are not appropriate as the perturbation $\delta u$ and $\delta y$ are unknown.  To solve it, let $q=[\delta u', \delta y']'\in \bR^{2n}$. From the worst point of view, $\theta$ is obtained by solving a RCO problem
\bee
\min_{\theta, t} t ~\text{subject to}~\twon{(y+\delta y)-(U+\delta U)\theta}\le t, \forall q\in\cQ. \label{sol_sp}
\ene

If $\cQ=\{q|\twon{q}_{\infty}\le \rho\}$, where $\rho$ represents the uncertainty size, it is a structured robust LS problem, which is NP-complete \cite{el1997robust}. Thus, we approximately solve it by the scenario approach via distributed Algorithms \ref{alg_dist} and \ref{alg_dist1}, and set $u=[1~2~3]'$ and $y=[4~5~6]'$. While for the SP,  we consider $\epsilon=0.002$ and $\delta=10^{-4}$. This implies from (\ref{scenariosize}) that $N_{bin}\ge 8868$. Here we set $N_{bin}=10000$   and each node  independently extracts samples via a uniform distribution over the uncertainty set $\cQ$. 

 We adopt three types of undirected graphs, see Fig. \ref{fig_graph} where the random graph is obtained by further connecting node $i$ to node $j (\neq i+1)$ with probability $p=0.2$ in a cycle graph.
A directed random graph is originated from the undirected one. Specifically, node $i$ is connected to node $i+1$ in the clockwise direction, and the direction of every other link is randomly selected with equal probability.
 
\begin{figure}[!t]	
		\centering
		\includegraphics[width=8.5cm,height=3cm]{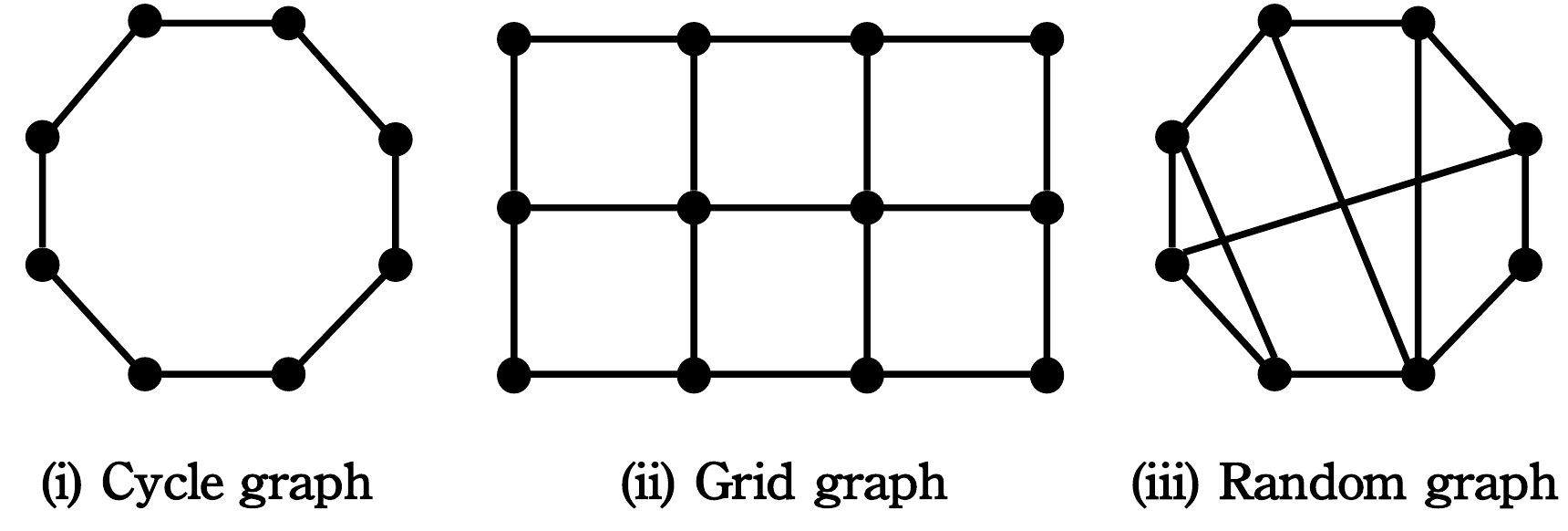}
		\caption{Three types of graphs.}
		\label{fig_graph}	
\end{figure}

 Given an uncertainty size, define the maximum of the scenario-based residuals by
\bee
r(\theta,\rho)=\max_{i=1,\ldots,N_{bin}}\twon{(y+\delta y^{(i)})-(U+\delta U^{(i)})\theta}.\label{residual}
\ene
Let $\theta_{ls}=U^{-1}y$ be the solution of the standard LS  and $\theta_{sc}$ be solution to the SP of (\ref{sol_sp}), which is computed by Algorithm \ref{alg_dist}. We depict the maximum residuals of (\ref{residual}) in Fig. \ref{fig_error} under different sizes of uncertainty, which shows the robustness of the solution of the SP.   Then, we compare the convergence behavior of the proposed algorithms for the SP of (\ref{sol_sp}) with $N_{bin}=10000, \rho=0.2$ and $\zeta^k=2/k$ in (\ref{stepsize}).   Fig. \ref{fig_iteration} shows that Algorithm \ref{alg_dist} converges to a solution of the SP much faster than that of Algorithm \ref{alg_dist1}.

\begin{figure}[!t]	
	\centering
	\begin{subfigure}[t]{0.5\linewidth}
		\centering
		\includegraphics[width=\textwidth,height=3.75cm]{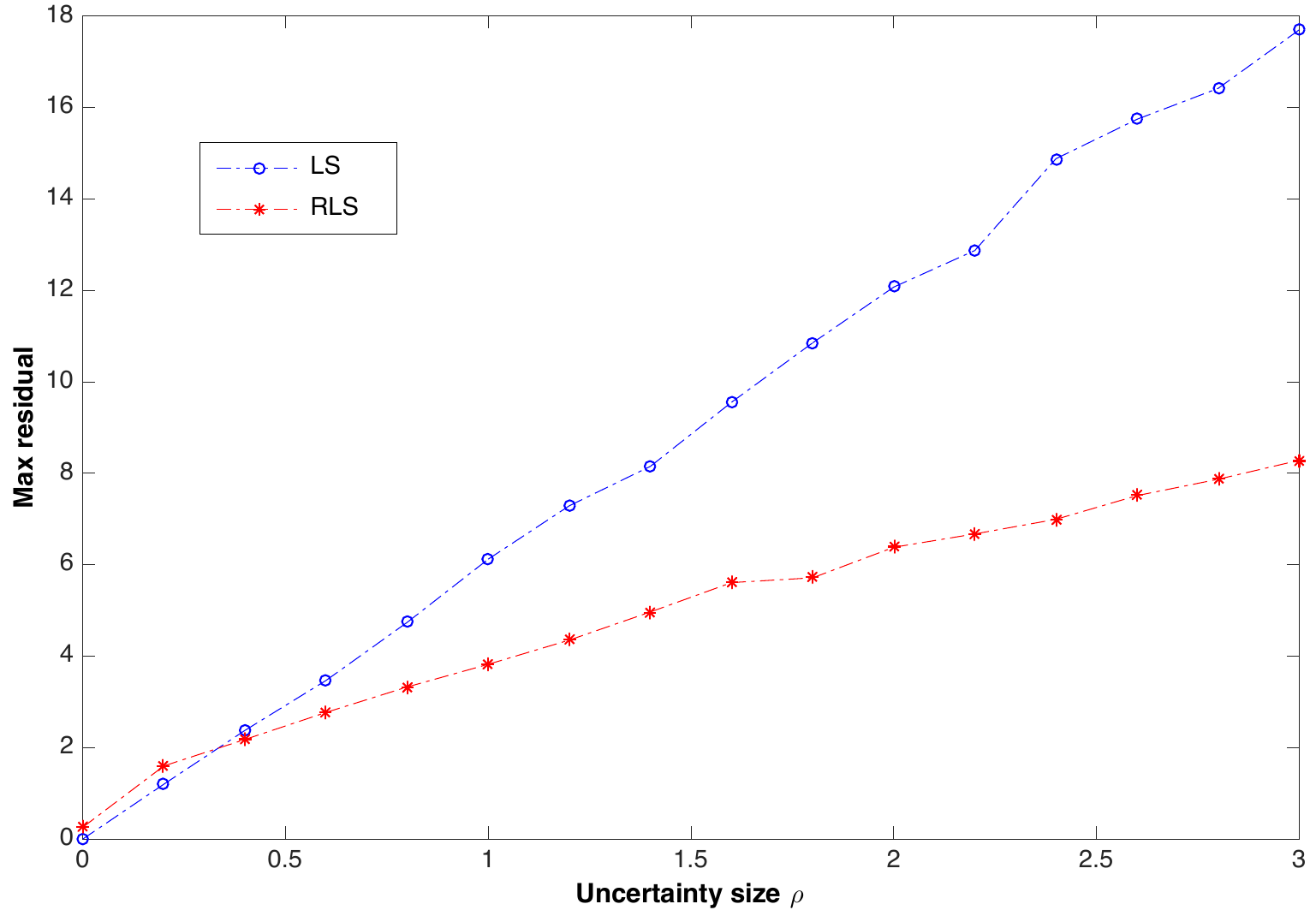}
		\caption{}
		\label{fig_error}	
	\end{subfigure}%
	\begin{subfigure}[t]{0.5\linewidth}
		\centering
		\includegraphics[width=\textwidth,height=3.75cm]{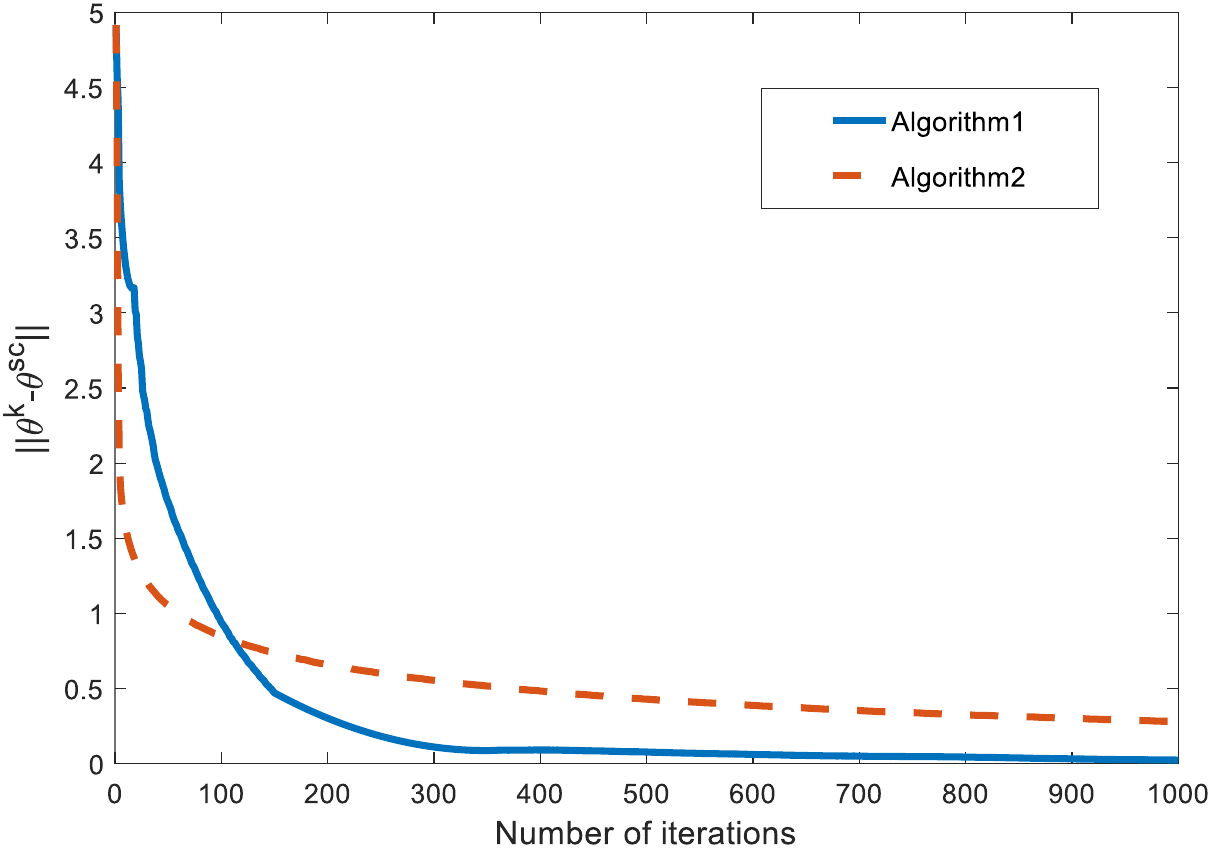}
		\caption{}
		\label{fig_iteration}
	\end{subfigure}
	\caption{(a) Maximum residual versus uncertainty size. (b) Convergence behaviors of Algorithms \ref{alg_dist} and \ref{alg_dist1} with $\beta=1.5$ on  undirected and directed random graphs with $m=100$.}
	\label{fig2:exp2}
\end{figure}

Since for both algorithms the dimensions of the data in crossing a communication link and being stored and retrieved in local memory are constant, one can argue that the total time to run our algorithms is essentially given by  
$$
T_{total}=T_{comp}+\alpha\cdot N_{iter},
$$
where $T_{comp}$ is the time attributed just to computation,  $\alpha$ is a constant which mainly depends on the network topology, the communication protocol and the  memory access speed, and $N_{iter}$ denotes the number of iterations. Let $T_i^k$ be the time cost to compute the $k$-th iteration of node $i$, i.e., running steps 3 and 4 in Algorithm \ref{alg_dist}.  Then, it follows from \cite[Section 1.2.2]{bertsekas2003parallel} that $T_{comp}=\sum_{k} \max_{i\in\cV}\{T_i^k\}$. Fig. \ref{fig_timevsnode} illustrates how the number of nodes affects $T_{comp}$, which decreases rapidly if the node number $m$ is small, and is indistinguishable for three types of network topologies as each node only involves simple numerical operations. This is consistent with our objective to reduce the computation cost of each node. Moreover, Fig. \ref{fig_timevsnode} also indicates that $T_{comp}/m$ is uniformly bounded away from zero, showing the practicability of the proposed distributed algorithm \cite[Section 1.2.2]{bertsekas2003parallel}. Ideally, $T_{comp}/m$ needs to be a constant, which is however not attainable \cite[Section 1.2.2]{bertsekas2003parallel}. 

Fig. \ref{fig_density} illustrates that the graph with denser communication links requires a smaller number of iterations, which is clearly consistent with our intuition as the information is  mixing faster over a denser graph. However, this requires a higher communication cost. By Fig. \ref{fig2:exp2}, one can conclude that designing an optimal topology is extremely complicated, and requires an optimal tradeoff among the communication topology, the number of nodes, and the computation and storage capacity of a single node, some of which are highly coupled. Similar phenomenon can be observed for Algorithm \ref{alg_dist1} and is not included to save space.
\begin{figure}[!t]	
	\centering
	\begin{subfigure}[t]{0.5\linewidth}
		\centering
		\includegraphics[width=\textwidth,height=3.75cm]{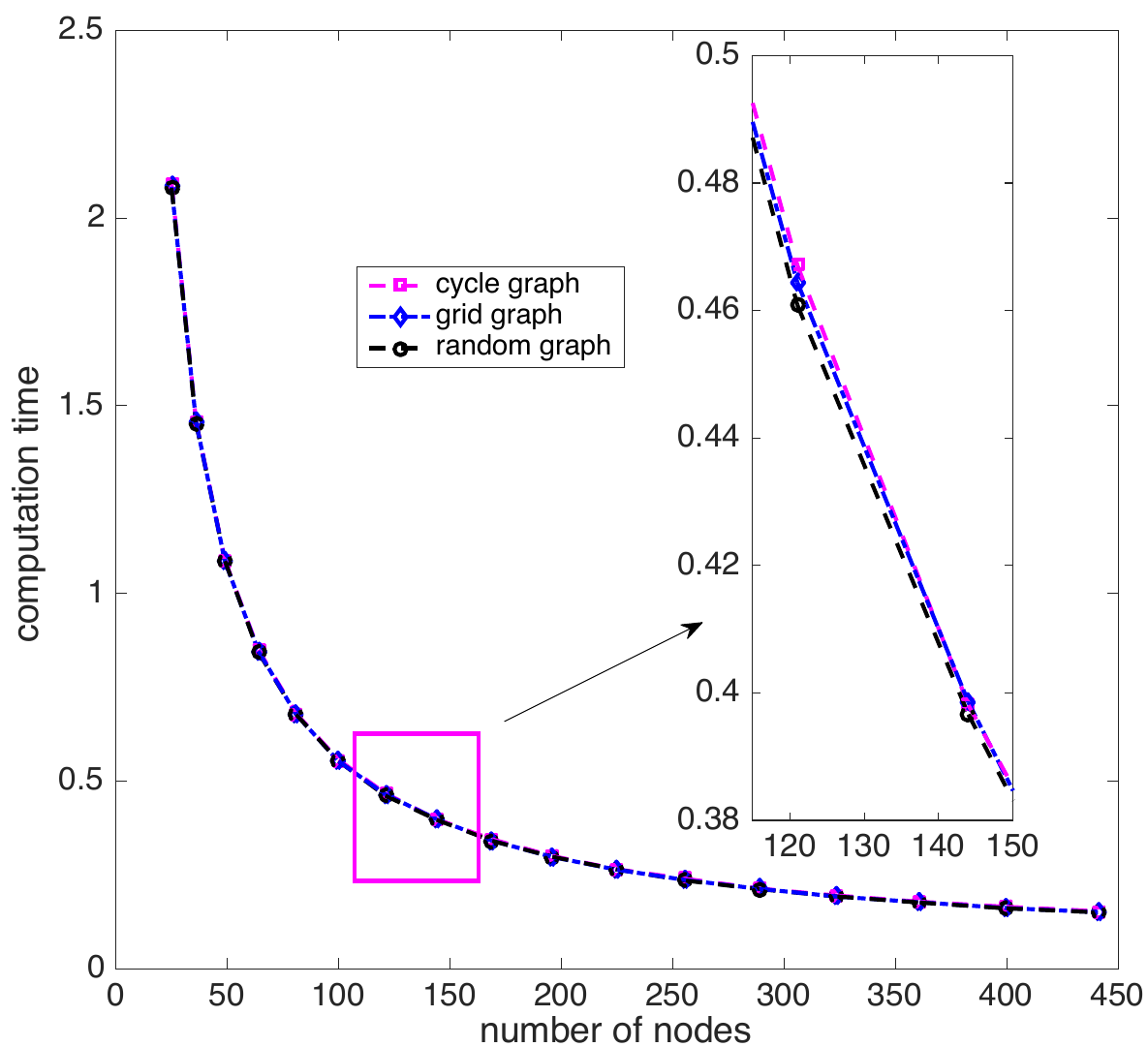}
		\caption{}
		\label{fig_timevsnode}	
	\end{subfigure}%
	\begin{subfigure}[t]{0.5\linewidth}
		\centering
		\includegraphics[width=\textwidth,height=3.75cm]{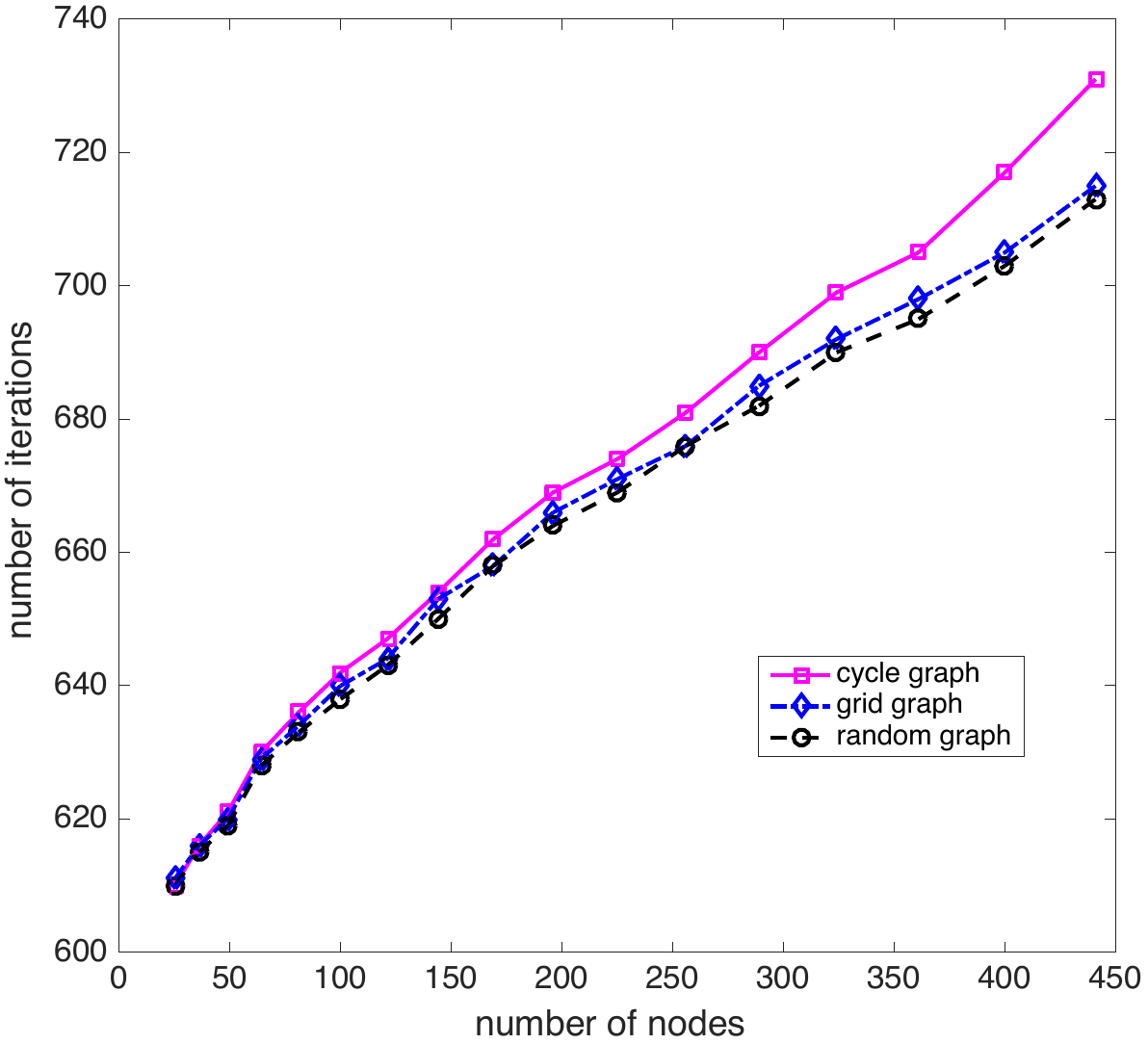}
		\caption{}
		\label{fig_density}
	\end{subfigure}
	\caption{Performance of Algorithm \ref{alg_dist} over  three types of network topologies. (a) The time (second) attributed just to computation versus the network size. (b) The number of iterations versus the network size.}
	\label{fig2:exp2}
\end{figure}

\section{Conclusion}
\label{sec_conclusion}

In this work, we developed distributed algorithms to collaboratively solve RCO via the SP, which possibly has a large number of constraints. Two distributed algorithms with very simple structure were provided for undirected and directed graphs, respectively. Compared with the existing results, the complexity per iteration of the proposed algorithms is significantly reduced.  Future work will focus on exploiting the structure of the parametrized constraint functions to reduce the computation cost.

\section*{Acknowledgement}
The authors would like to thank the Associate Editor and anonymous reviewers for their very constructive comments, which greatly improve the quality of this work.

\bibliographystyle{IEEEtrans}
\bibliography{mybibf}
\begin{IEEEbiography}
[{\includegraphics[width=1in,height=1.25in,clip,keepaspectratio]{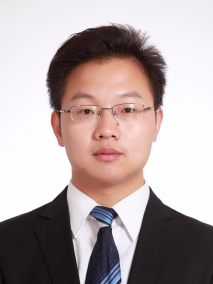}}]
{Keyou You}   received the B.S. degree in Statistical Science from Sun Yat-sen University, Guangzhou, China, in 2007 and the Ph.D. degree in Electrical and Electronic Engineering from Nanyang Technological University (NTU), Singapore, in 2012. After briefly working as a Research Fellow at NTU, he joined Tsinghua University in Beijing, China where he is now an Associate Professor in the Department of Automation. He held visiting positions at Politecnico di Torino, The Hong Kong University of Science and Technology, The University of Melbourne and etc. His current research interests include networked control systems, distributed  algorithms, and their applications.

Dr. You received the Guan Zhaozhi award at the 29th Chinese Control Conference in 2010, and a CSC-IBM China Faculty Award in 2014. He was selected to the National 1000-Youth Talent Program of China in 2014, and received the National Science Fund for Excellent Young Scholars in 2017.\end{IEEEbiography}

\begin{IEEEbiography}
[{\includegraphics[width=1in,height=1.25in,clip,keepaspectratio]{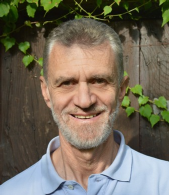}}]
{Roberto Tempo}   was a Director of Research of Systems and Computer Engineering at CNR-IEIIT, Politecnico di Torino, Italy. He held visiting positions at Tsinghua University  in Beijing, Chinese Academy of Sciences, Kyoto University, The University of Tokyo, University of Illinois at Urbana-Champaign, German Aerospace Research Organization in Oberpfaffenhofen and Columbia University in New York. His research activities were focused on the analysis and design of complex systems with uncertainty, and various applications within information technology.  

Dr. Tempo was a Fellow of the IEEE and a Fellow of the IFAC, a recipient of the IEEE Control Systems Magazine Outstanding Paper Award, of the Automatica Outstanding Paper Prize Award, and of the Distinguished Member Award from the IEEE Control Systems Society. He was a Corresponding Member of the Academy of Sciences, Institute of Bologna, Italy, Class Engineering Sciences. In 2010 Dr. Tempo was President of the IEEE Control Systems Society. He served as Editor-in-Chief of Automatica, Editor for Technical Notes and Correspondence of the IEEE Transactions on Automatic Control and Senior Editor of the same journal.\end{IEEEbiography}

\begin{IEEEbiography}
[{\includegraphics[width=1in,height=1.25in,clip,keepaspectratio]{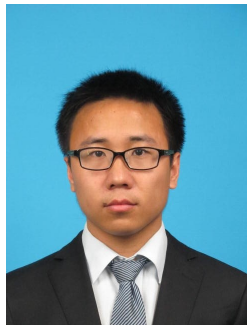}}]
{Pei Xie} received the B.E. degree from the Department of Automation, Tsinghua University, Beijing, China, in 2013. Currently, he is working toward his Ph.D. degree of the same institute. His research interests include networked control system, distributed optimizations, and their applications.
 \end{IEEEbiography}

\end{document}